\documentclass[magisterska]{pracamgr}
\usepackage{graphicx}
\usepackage[latin2]{inputenc}
\usepackage{graphicx}
\usepackage{amssymb}
\usepackage{amsthm}
\usepackage{amsmath}
\usepackage{amsfonts}
\usepackage{bbold}
\usepackage{subfig}
\usepackage{url} 
\usepackage{multirow}
\author{Tomasz Badowski}
\nralbumu{260542}
\title{Variance-based sensitivity analysis for stochastic chemical kinetics}

\tytulang{Oparta o wariancj\c{e} analiza wra\.zliwo\'sci stochastycznych modeli reakcji chemicznych}

\kierunek{Computer Science} 

\opiekun{dr hab. Anna Gambin\\
  Institute of Informatics 
}

\date{Warsaw, Semptember 2011}

\dziedzina{ 
11.2 Statistics\\ 
}

\klasyfikacja{65C05   	Monte Carlo methods}


\keywordsEng{variance-based sensitivity analysis, stochastic simulations, Gillespie's direct method,
 Monte Carlo, Markov process, chemical reaction}

\streszczang{Analiza wra\.zliwo\'sci to proces obliczania wsp\'o\l{}czynnik\'ow wra\.zliwo\'sci, kt\'ore s\c{a} pewnymi miarami wa\.zno\'sci 
parametr\'ow pod wzgl\c{e}dem ich wp\l{}ywu na wyniki modeli matematycznych. 
Wsp\'o\l{}czynniki wra\.zliwo\'sci obliczane w analizie wra\.zliwo\'sci opartej o wariancj\c{e} dostarczaj\c{a} ilo\'sciowych 
odpowiedzi pytania takie jak np. o ile \'srednio zmniejszy si\c{e} wariancja wynik\'ow modelu, mierz\c{a}ca ich niepewno\'s\'c,
 je\'sli wyznaczymy dok\l{}adne warto\'sci 
niekt\'orych nieznanych parametr\'ow, np. do\'swiadczalnie. Proponujemy nowe schematy do estymacji opartych o 
wariancj\c{e} wsp\'o\l{}czynnik\'ow wra\.zliwo\'sci wynik\'ow modeli stochastycznych, ich warunkowych warto\'sci oczekiwanych i histogram\'ow wzgl\c{e}dem parametr\'ow. 
Nieobci\c{a}\.zone estymatory otrzymywane w tych schematach mog\c{a} zosta\'c wykorzystane w procedurze Monte Carlo (MC) aproksymuj\c{a}cej 
wsp\'o\l{}czynniki wra\.zliwo\'sci. 
Wyznaczamy relacje mi\c{e}dzy wariancjami ko\'ncowych estymator\'ow procedur MC wykonuj\c{a}cych tyle samo ewaluacji danej funkcji ale 
u\.zywaj\c{a}c r\'o\.znych schemat\'ow, zar\'owno dla nowo podanych schemat\'ow, jak i pewnych dotychczas u\.zywanych w literaturze. 
Eksperyment numeryczny dla dyskretnego Markowowskiego modelu uk\l{}adu reakcji chemicznych (DM)
 pokazuje, \.ze nasza metoda mo\.ze prowadzi\'c do o wiele 
mniejszych b\l{}\c{e}d\'ow ni\.z metoda zaproponowana przez Degasperiego i innych \cite{Degasperi2008}. Dalsze eksperymenty numeryczne pokazuj\c{a}, \.ze 
u\.zycie algorytmu random time change (RTC) zaproponowanego przez Rathinama i innych do symulacji DM mo\.ze prowadzi\'c do ponad 
30-krotnie mniejszej wariancji estymator\'ow pewnych wsp\'o\l{}czynnik\'ow wra\.zliwo\'sci, ni\.z metoda Gillespie's direct (GD) i \.ze 
wariancja ta mo\.ze si\c{e} bardzo zmienia\'c wraz ze zmian\c{a} kolejno\'sci reakcji w metodzie GD. Dostarczamy pewnych intuicji wyja\'sniaj\c{a}cych 
te efekty. Uog\'olniamy miary s\l{}u\.z\c{a}ce do por\'ownywania rozrzutu r\'o\.znych rozk\l{}ad\'ow, jak na przyk\l{}ad 
wsp\'o\l{}czynnik zmienno\'sci (ang. coefficient of variation), czy 
wsp\'o\l{}czynnik Fano na przypadek losowych parametr\'ow, w taki spos\'ob, \.ze mog\c{a} one by\'c obliczane r\'ownocze\'snie ze 
wsp\'o\l{}czynnikami wra\.zliwo\'sci opartymi o wariancj\c{e}.
Metody zaproponowane w tej pracy s\c{a} og\'olne i mog\c{a} zosta\'c zastosowane 
do analizy wra\.zliwo\'sci opartej o wariancj\c{e} modeli stochastycznych w r\'o\.znych dziedzinach. }

 \newcommand{\wt}[1]{\widetilde{{#1}}}
 \newcommand{\mc}[1]{\mathcal{{#1}}}

\newcommand{\N}{\mathbb{N}}
\newcommand{\R}{\mathbb{R}}
\newcommand{\I}{\mathbb{1}}
\newcommand{\PR}{\mathbb{P}}

\newtheorem{theorem}{Theorem}
 
\newtheorem{lemma}[theorem]{Lemma}
\newtheorem{defin}[theorem]{Definition}
\newtheorem{constr}{Construction}

\DeclareMathOperator{\Var}{Var}
\DeclareMathOperator{\RN}{RN}
\DeclareMathOperator{\hist}{hist}
\newcommand{\E}{\mathbb{E}}

\DeclareMathOperator{\U}{U}
\DeclareMathOperator{\Exp}{Exp}
\DeclareMathOperator{\Pois}{Pois}

\DeclareMathOperator{\Cov}{Cov}
\DeclareMathOperator{\err}{err}
\DeclareMathOperator{\dist}{dist}

\begin{document}
\maketitle
\begin{abstract}
Sensitivity analysis is a process of computing sensitivity indices, which are
certain measures of importance of parameters in influencing the outputs of mathematical models.
Sensitivity indices computed in variance-based sensitivity analysis yield quantitative answers to questions
like how much on average the variance of model output, measuring its uncertainty, decreases, if exact values of certain
unknown parameters are determined, e. g. in an experiment.
We propose new schemes for estimation of variance-based sensitivity
indices of outputs of stochastic models, their conditional expectations and histograms given the parameters.
Unbiased estimators obtained in these schemes can be used in a Monte Carlo (MC) procedure approximating sensitivity indices.
We derive relations between variances of final estimators of MC procedures making the same number of evaluations of given function, but using
different schemes, both for the newly introduced schemes and for some used before in the literature.
Numerical experiment for a discrete state stochastic Markov model of a chemical reaction network (DM) shows that our method
can lead to much lower error than method analogous to the one offered by Degasperi et al. \cite{Degasperi2008}.
Further numerical experiments demonstrate that the
application of random time change (RTC) algorithm due to Rathinam et al. for simulation of DM
can lead to over 30 times lower variance of estimators of certain sensitivity indices than when Gillespie's direct (GD) method is used, and
that this variance may significantly depend on the order of reactions in GD method.
We provide some intuitions explaining these effects.
We generalize measures used for comparing dispersion of different distributions, such as coefficient of variation and Fano factor
to the random parameters case, in a way that they can be computed along with variance-based sensitivity indices.
The methods proposed in this work are general and can be used for variance-based sensitivity analysis of stochastic models in any discipline.


\end{abstract}
\tableofcontents
\chapter*{Introduction}
\addcontentsline{toc}{section}{Introduction}
Mathematical models often take some parameters as inputs and return some results, which we call outputs of the model, and
which are certain functions of the inputs.
For instance in popular ordinary differential equations based models \cite{Atkins_2006} in chemical reaction kinetics the
parameters can be initial concentrations of reacting species and reaction rate constants, while the output
can be concentration of any species at a given time.
Many physical systems, such as chemical reaction networks involving small
concentrations of certain species, are well described by stochastic models \cite{Pahle2009, vanKampen_B07}. It is for instance becoming clear that
such models can successfully
describe the functioning of a number of important biochemical systems, including certain gene regulatory networks \cite{Arkin1999, Rao_Wolf_Arkin_2002}
and signaling pathways \cite{Lipniacki2007, Tay_2010}.
The output of a stochastic model with given parameters is usually not a single value but random variable with distribution specified by the
 parameters.
Well-stirred chemical reaction networks with small numbers of certain particles are often described using discrete stochastic
 Markov model (DM), the history of which is reviewed in \cite{mcquarrie}.
The inputs of DM can be reaction rates and some parameters describing
initial distributions of particles and the output might be for instance
the random number of particles at a given time.
A number of other stochastic formalisms of chemical kinetics
have also been used such as chemical Langevin equation or hybrid stochastic-deterministic models \cite{Pahle2009}, the latter being particularly
useful for modelling reactions with both small and large concentrations of different species.

Sensitivity analysis methods are concerned with computing different measures of relative importance of
arguments in influencing the value of a function, and in particular can be applied to outputs and inputs of mathematical models.
In stochastic models parameters of distribution of the output like mean \cite{Rathinam_2010}, variance \cite{Barmassa} or histograms
\cite{Degasperi2008} are often taken as functions 
 whose sensitivity indices with respect to model inputs are computed.
Sensitivity analysis has found applications in such diverse fields as chemical kinetics \cite{Rabitz_Kramer_Dacol_1983, Turanyi_1990,
Saltelli2005, Van_Riel_2006}, nuclear safety \cite{Iooss2008}, environmental science \cite{Tarantola_Giglioli_Jesinghaus_Saltelli_2002}
or molecular dynamics \cite{Cooke_Schmidler_2008}.
In chemical kinetics sensitivity analysis has been used among others for
parameter estimation \cite{Kim_Spencer_Albeck_Burke_Sorger_Gaudet_Kim_2010, Juillet2009},
and model simplification \cite{Cristaldi_2011, Okino_Mavrovouniotis_1998, Liu_Swihart_Neelamegham_2005, Degenring2004}.
Main types of sensitivity analysis methods are local and global ones.
Local sensitivity analysis methods deal with approximating derivatives of the function with respect to its arguments
in a given point of the parameter space. 
A number of attempts were made to speed up the approximation of these derivatives for DM \cite{Rathinam_2010, Plyasunov2007}.
In particular Rathinam et al. \cite{Rathinam_2010}
showed that using random time change (RTC) algorithm, which is based on the representation of discrete Markov processes due to Kurtz \cite{Kurtz1986},
may lead to much lower variance of estimators of finite difference of mean values of particle numbers at a given time,
 than when Gillespie's direct (GD) method is used.

In global sensitivity analysis the arguments of a function are considered to be random variables. 
They might be for instance results of uncertain measurements of some rate constants.
Global sensitivity analysis methods provide measures of importance of parameters
in influencing the value of a function over the whole range of their possible values.
Variance-based sensitivity analysis is a well established type of global sensitivity analysis, 
which has a long history of being used in chemical kinetics, its
first formulation known as FAST method having been introduced for this purpose in the seventies \cite{Saltelli2008, Cukier1973}.
Importance indices provided by variance-based sensitivity analysis yield quantitative answers to questions which might serve as reasons for
 undertaking the analysis.
Some of these questions are what average reduction of variance and thus improvement of precision of the model
can be achieved if some of the unknown constants are determined by a precise measurement \cite{Saltelli2008}, or what average error is caused by
fixing a parameter for instance to simplify the model \cite{Sobol2007}. 
Variance-based sensitivity analysis has been used among others for parameter estimation in a linear compartmental model \cite{Juillet2009},
and was demonstrated useful for reducing a model of a certain stage of production of an anti parasitic
medicine Ivermectin \cite{Cristaldi_2011}.

The only work so far, in which variance-based sensitivity analysis was performed for parameters of distribution of outputs of
stochastic models, was a paper by Degasperi et al. \cite{Degasperi2008}. Authors performed simulations in every point
of a grid in the parameter space to approximate conditional histograms given the parameters. Unfortunately,
the generalization of variance to the case of vector-valued functions they used causes basic properties of variance-based sensitivity indices, crucial for
their certain applications, to be lost.
Furthermore, their method provides no error estimates for the results.

We propose how to generalize variance to vector-valued functions, like conditional histograms, so that interpretations of variance-based sensitivity indices hold.
We introduce the concept of a scheme for estimation of sensitivity index, containing among others the information of an unbiased estimator
for the index and the number of function evaluations needed to compute it.
We propose different schemes for estimation of variance-based sensitivity indices of conditional expectations and histograms of outputs of
stochastic models given the parameters, which can be used to compute estimators in each step of Monte Carlo (MC) procedure. Thus we can obtain
not only estimates of the indices, but also estimates of error of the result, that is of variance of final MC estimator.
We introduce inefficiency constant of a scheme in estimating given index equal to variance of estimator given by the scheme
times the number of function evaluations used to compute it.
Ratio of such constants for two different schemes is equal to the ratio
 of variances of final MC estimators using
these schemes and the same number of evaluations of a given function with certain distribution of its parameters. Thus if function
evaluations are the main cost of MC step, as is in case of our experiments, these constants can be used to compare error resulting from using different
 schemes for the same computation time. We derive relations between inefficiency constants for different schemes,
both the ones already used to deterministic chemical models
in the literature and the ones introduced in this work.

Numerical experiments on example for which analytic values
of sensitivity indices can be obtained demonstrate better performance of our method in comparison to a method analogous to
Degasperi's et al. in computing sensitivity indices of conditional means, and the fact that
 quasi-Monte Carlo can lead to speed-up in computation of some indices. Further numerical experiments demonstrate that
using random time change (RTC) algorithm can lead to lower variance of certain estimators computed in our schemes, than when
Gillespie's direct method (GD) is used. We also show that this
variance is influenced by order of reactions used in GD method and give some intuitions concerning possible reasons for this dependence.
 Along with sensitivity indices we compute newly introduced generalizations
of measures of dispersion of distribution of outputs of stochastic models to the random parameters case.

The rest of this work is organized as follows. In Chapter \ref{discrStoch} we
define chemical reaction network, define DM for the constant parameters case and provide its constructions.
In Chapter \ref{randPar} we discuss possible interpretations of random parameters in models and extend DM construction to the random parameters case.
In Chapter \ref{varBased} we define variance-based sensitivity indices and describe their interpretations as well as
possible applications.
In Chapter \ref{MC} we define and provide schemes for estimation of sensitivity indices, define their inefficiency constants and
 derive relations between them.
We also discuss implementation details, introduce method analogous to Degasperi's et al. and method depending on using quasi-Monte Carlo sampling
in the parameter space. Chapter \ref{chapNumExp} is devoted to numerical experiments.

We assume basic knowledge of the reader about probability theory, such as contained in the first Chapter of  \cite{Durrett}.
More advanced definitions and theorems as well as assumptions we are making throughout the text are given or referenced to in the main text
or Appendix \ref{appMath}.

\chapter{Stochastic chemical models with constant parameters}\label{discrStoch}
\section{Chemical reaction network}\label{CRN}
A chemical reaction network contains $N$ chemical species with symbols $X_1,...,X_N$. In DM formalism the
state of the system at a given moment is characterized by a vector of natural numbers
$x = (x_1,\ldots,x_N)$ from some set of admissible states $E \subset \N^N$.
Coordinates of $x$ describe the numbers of species of each kind.
$L$ chemical reactions $(R_1,\ldots,R_L)$ can occur. The $l$-th reaction is described by formula
\begin{eqnarray}\label{reac}
\underline{s}_{l,1} X_{1} + ... + \underline{s}_{l,N} X_{N} \longrightarrow  \overline{s}_{l,1} X_{1} + ... + \overline{s}_{l,N} X_{N}.
\end{eqnarray}
Vector $\underline{s}_l = (\underline{s}_{l,i})_{i=1}^{N}$
is called the stoichiometric vector of reactants and $ \overline{s}_l = (\overline{s}_{l,i})_{i=1}^{N}$ of products
of reaction $R_l$.
In all this work we denote $I_n = \{1,\ldots, n\}$.
We require that $\underline{s}_l \geq 0 $, which means $\underline{s}_{l,i} \geq 0$  for $i \in I_N$ and similarly $\overline{s_l} \geq 0$.
We define transition vector as $s_{l} =  \overline{s_l} -  \underline{s_l}$.
Occurrence of $l$-th reaction makes the system at state $x$ to transition to state $x +  s_l$.
With every reaction $R_l$ is associated a propensity $a_l(k)(x)$ -
a non negative
function of state $x \in E$ and real vector $k$, called (vector of) rate constants, which can have values in some set $S_k \subset \R^{n_k}$
for certain $n_k$ natural positive.
Intuitively speaking, propensity describes how quickly reaction is proceeding in state $x$.
For propensities in DM
we require that $a_l(k)(x)  = 0$ if for some $i\in I_n$ $x_i < \underline{s}_i $, that is if there are too few particles of certain reactant in
the system for the reaction to happen. For instance in the stochastic version of mass action kinetics \cite{Kurtz1986} we take
$k = (k_i)_{i=1}^{L}$ and
\begin{eqnarray}
a_l(k)(x) = k_l{ x \choose \underline{s}_{l}},
\end{eqnarray}
where ${x \choose \underline{s}_{l} } = \prod_{i=1}^N{x_i \choose \underline{s}_{l,i} }$ is the number of
possible ways in which the reactants can collide for the $l$-th reaction to occur, and $k_l$ is called the rate constant of this reaction.
Formally, we describe chemical reaction $R_l$ as a triple, which is a function of the rate constants
\begin{equation}
R_l(k) = (a_l(k)(x), \underline{s}_l, \overline{s}_l),
\end{equation}
 and chemical reaction network RN as a tuple containing sequence of reactions and the set of species
\begin{equation}\label{NTfun}
\RN(k) = \big((R_l)_{l=1}^L, \{X_1,..., X_N\}\big),
\end{equation}
also being a function of the rate constants.
A reverse reaction to a given is one in which stoichiometric vectors of reactants and products are replaced. We say that reaction
is reversible if both the reaction and its reverse are present in the reaction network.

\section{Discrete stochastic model with constant parameters}
We define stochastic chemical reaction network with constant parameters $p = (c,k)$ (DMCP)
as a certain Markov process on $E$ with allowed times $T = [0, \infty)$, which is a
type of right-continuous stochastic process \cite{Norris1998}. 
Shortly, stochastic process $Y$ with values in $E$ and allowed times $T$ is a family of random variables $(Y_t)_{t \in T}$ with values in $E$.
One can treat $Y$ as a random variable whose values, known as trajectories of the process, are elements of $E^T$   \cite{billingsley1979}.
Right-continuous process is one that behaves as follows. It starts in some state $Z_0 \in E$, where it waits for time period 
$S_0$ at the end of which it transitions to another state $Z_1$, where it waits for another time $S_1$ and so on for some discrete process $(Z_i)_{i\geq 0}$,
known as jump process of $Y$ and $(S_i)_{j \geq 0}$ known as its holding times of $Y$ (see \cite{Norris1998} for precise definitions).
The moment 
\begin{equation}\label{expTime}
\chi = \sum_{i=0}^{\infty}S_n,
\end{equation}
when a process makes infinitely many jumps for the first time is called its explosion time.
If $\chi = \infty$, that is no infinite number of transitions can occur in finite time, we call the process nonexplosive.
Since in our case transition corresponds to firing of a reaction and it is nonphysical for infinite number of reactions to occur in finite time
we require the process to be nonexplosive.
A right-continuous nonexplosive stochastic process is uniquely defined by its jump chain and holding times.
An important property of stochastic process $Y$ is its distribution defined similarly as for random variable \cite{billingsley1979}.
The distribution of a Markov process is defined by its distribution at time 0 and 
non negative numbers $q_{ x,y}$ for ${ x, y} \in E$, ${x} \neq {y}$ known as intensities of going from state
$x$ to $y$, fulfilling
\begin{equation}
\forall{x \in E}\quad \sum_{y \in E}\ q_{ x,y} < \infty.
\end{equation}
By DMCP with parameters $p = (c,k)$ and corresponding to a given chemical reaction network RN we mean
 a nonexplosive Markov process on $E$ with deterministic initial distribution $\delta_{c}$ and 
 intensities for $x,y \in E$, $x \neq y,$ equal to
\begin{equation}\label{qxyEq}
 q_{ x,y} = \sum_{l: \  y = x + s_l}\ a_l(k)(x).
\end{equation}
Unfortunately, not for all reaction networks and values of parameters $p$ DMCP exists \cite{KurtzReview2010}.
We give some sufficient conditions in the next Section. 
The vector of parameters  $p = (k, c)$ of DMCP corresponding to a given chemical reaction network RN
uniquely determines its distribution, which we denote $\mu_{DMCP}(RN(k), c)$. 
One can model reaction networks using other types of processes whose distributions can also be specified using certain parameters.
For instance for the chemical Langevin equation \cite{Wilkinson2006} such parameter 
 vector $p =(c,k)$ would contain initial species concentrations $c$ instead of species numbers.
One can also consider models incorporating different types of events
during the simulation, whose distribution depends on some additional parameters characterizing
these events.
For instance in the numerical simulations of the stimulation of NF-$\kappa$B regulatory network in \cite{Lipniacki2007} with
tumor necrosis factor-alpha (TNF-$\alpha$) one could consider the dose of TNF-$\alpha$ used for stimulation or the time when the stimulation begins
as such additional parameters.

\section{Constructions of DMCP}\label{secConstr}
We show two possible constructions of DMCP with constant parameters $p = (c,k)$ corresponding to a given reaction network RN (\ref{NTfun}),
assuming that any such process exists.
The first one is based on GD method introduced in \cite{Gillespie1976} and the second is based on RTC algorithm introduced
in \cite{Rathinam_2010} and is equivalent to random time change representation of Markov processes due to Kurtz \cite{Kurtz1986}.
In both constructions we inductively define the jump chain $(Z_n)_{n \geq 0} $ and holding
times $(S_n)_{n \geq 1}$.  
Notation $X \sim \text{U}(0,1)$ means
random variable $X$ has distribution  U($0,1$), which in this case means uniform on the interval $[0,1]$. Exp($1$) means exponential distribution with
parameter $1$ \cite{Norris1998}.
Notation $X \sim Y$ means that random variables $X$ and $Y$ have the same distribution.
\begin{constr}[GD construction]\label{GCon}
Let $U_1,U_2,\ldots $ be independent identically distributed (i. i. d) random variables,  $U_1 \sim \U(0,1)$,
 and $T_1,T_2,\ldots$ i. i. d. $T_1 \sim \Exp(1)$.
Let us assume that $Z_i$, $S_i$ are defined for some $i \geq 0$. We set
\[ q: = \sum_{l=1}^L a_l(k)(Z_i). \]
If $q = 0$, then we place
\[ S_{i+1} := \infty,\ Z_{i+1} := Z_{i}.\]
Otherwise, we set
\[S_{i+1} := \frac{T_{i+1}}{q}\]
and for
\[l = \min\{m \in I_{L}: \frac{1}{q}\sum_{n=1}^{m} a_{n}(k)(Z_i) \geq \U_{i}\} \]
we place
\[Z_{i + 1} = x + s_l.\]
\end{constr}
For a given chemical reaction network and value of rate constants $k$ we denote
$A(k)(x) = \{l \in I_L:\ a_l(k)(x) > 0\} = \{l_1, \ldots, l_{L(x)}\}$ - the
set of $L(x)$ indices of reactions which can occur in state $x$.
\begin{constr}[RTC construction]\label{RTCCon}
Let us consider $L$ independent Poisson processes $(N_l)_{l=1}^L$ with unit rates.
The second construction tries to find the solution of the following integral equation
\begin{equation}\label{intEqu}
Y_t =  c + \sum_{l=1}^L s_l N_l(\int_0^t \! a_l(k)(Y_s) \, \mathrm{d}s).
\end{equation}
Let the $i$-th call of function $N_l.next$ return the $i$-th holding time of the Poisson process $N_l$.
We set for $l\in I_L$
\[ \tau_{0,l} := N_l.next. \]
Let us assume that  $Z_{i}$, $S_{i}$ and $\{\tau_{i ,l} \}_{l \in I_L}$ for some $i \geq 0$ were already defined.
We set
\begin{equation}
 S_{i+1} := \min_{l \in A(k)(Z_i)} \left\{ \frac{\tau_{i,l}}{a_l(k)(Z_i)}  \right\}.
\end{equation}
For a certain $l$ realizing the above minimum, we place
\begin{equation}
 Z_{i+1} := Z_i + s_l,\quad \tau_{i+1,l} = N_l.next.
\end{equation}
For $m \in A(k)(Z_i),\ m \neq l$, we place
\[ \tau_{m,i+1} := \tau_{m,i} - a_l(k)(Z_i) S_{i+1} \]
and for the remaining reaction indices $l \notin A(k)(Z_i)$ we set
\[ \tau_{l,i+1} := \tau_{l,i}. \]
\end{constr}
Above constructions define the process up to explosion time $\chi$ (\ref{expTime}). When $\chi < \infty$ for the jump times defined
by any of the above constructions, then we replace trajectory of the process by a trajectory constantly equal to some
$c_1 \in E$, so that we receive a nonexplosive right-continuous process.
Such constructed process is DMCP only if probability of event $\chi < \infty$, which is always equal for the above constructions, is also equal to $0$.
We then say that reaction network RN and parameters $p$ admit DMCP. An easy criterion for RN and $p$ to admit DMCP is given by the
following Theorem.
\begin{theorem}\label{thDMCP}
Using notations as in Section \ref{CRN} let us assume that for a reaction network RN and parameters $p = (c,k)$
there exist vector  $m =(m_i)_{i=1}^N \in \R^N$ with positive coordinates, such that for
$L_m = \{l \in I_L: s_{l}m > 0\}$ it holds 
\begin{equation}
A = \sup\{a_l(k)(x): x  \in E,\ l \in L_m\} < \infty,
\end{equation}
where by $s_{l}m$ we mean standard scalar product of vectors.
Then RN and $p$ admit DMCP.
\end{theorem}
Vector $m$ can be often taken to be vector of masses of each species, hence the notation.
\begin{proof}
 From continuity from below (see \cite{Durrett} Chap. 1 Ex. 1.1)
 it is sufficient to show that $\PR(\chi < t) = 0$ for every $t>0$. Let us consider
a helper process $Y_t$, which is created by running RTC construction with initial state $c$ and after the explosion
setting the state of the process to some vector $c_1$. Similarly as in the proof of Theorem 2.7.1 in \cite{Norris1998} one shows that
if $\PR(\chi < t) >0$, then $Y_t$ should take infinite number of values before time $t$ with nonzero probability. But
 since for every $M > 0$ the set $\{x \in \N^N: mx < M\}$ is finite
then also process $M_t$ defined as
\begin{equation}
M_t := mY_t
\end{equation}
should take infinitely many values before time $t$ with nonzero probability.
Denoting
\begin{equation}
s_m = \max\{s_lm :l \in I_L  \}
\end{equation}
we have from (\ref{intEqu})
\begin{equation}\label{ineqPois}
M_t = m Y_t \leq m(c+c_1) + \sum_{l=1}^L s_m N_l(tA) = m(c+c_1) + s_m N_{LtA},
\end{equation}
where $N_{LtA}$ is certain Poisson process with rate $LtA$ \cite{Norris1998}.
Since Poisson processes take finite number of values in finite time with probability $1$, the Theorem is proved.
\end{proof}
From now on, we consider the step of rejecting trajectories for which $\chi < \infty$ and replacing it by some arbitrary constant
from $E$ to be integral part of the above constructions.  
%
Note that in all constructions of processes used for computer simulations one
uses some collection of random variables $R$ to generate the random trajectories of the process.
For example for the first construction of $DMCP$ we have $R_1 = (U_i, T_i)_{i \geq 0}$, while for the second one $R_2 =  (N_i)_{i=1}^L$.
We call $R$ artificial noise variable, since it is a stochastic process which represents no physical quantity and
may be even different for different constructions of the same model.
Using certain construction of a process one can define function $h$ for which
\begin{equation}\label{DMCPFun}
h(p,R)
\end{equation}
is a stochastic process created by this construction with parameters $p$ and
 the artificial noise $R$. For all values of parameters processes given by constructions $1$ and $2$
have the same distributions, which can be expressed using their respective functions $(h_i)_{i=1}^2$ and artificial noise terms $(R_i)_{i=1}^2$ as
\begin{equation}\label{equivMarkov}
h_1(p,R_1) \sim h_2(p, R_2) \sim \mu_{DMCP}(RN(k),c).
\end{equation}
There are also other constructions equivalent to GD method, such as Gillespie first reaction method \cite{Gillespie1976},
Gibson and Bruck's next reaction method \cite{Gibson2000}, for which (\ref{equivMarkov}) also holds, but for a different function $h$
or artificial noise variable.



\chapter{Models with random parameters}\label{randPar}
\section{Random parameters}
There are many situations when we may want to treat the parameters of models as random
variables $P=(P_1,\ldots,P_N)$, rather than constants.
These variables can for instance represent uncertain quantities. One
often distinguishes 2 types of such variables (see \cite{deRocq_2008} and  \cite{Helton2003}
sec. 7.1 for more detailed descriptions and reviews of history of this distinction).
\begin{itemize}
\item Stochastic or aleatory: they are changeable in the model, like initial numbers of particles
of a species in the equilibrium distribution of a stochastic model. 
The uncertainty associated with distribution of this variable, measured for instance by its variance,
is also known as irreducible \cite{Helton2003}, since it cannot be reduced
by gaining further knowledge about the model. 
\item Epistemic: they are also known as state of knowledge or subjective \cite{Helton2003},
 since their distribution represents modeller's best judgement
about their possible values. Reaction rates can often be considered to be of this kind.
The judgement can be based on different values available in the literature \cite{Juillet2009}
or on the fact that model with parameters from given range well describes certain experimental data \cite{Schaber2009}.
Uncertainty associated with distribution of these variables has been called 
reducible \cite{Helton2003}, since it can be reduced if we gain more knowledge about the model, e. g. we can get to know
the values of rate constants by measuring them.
\end{itemize}
Different parameters may need to be considered not independent for a given model to be realistic.
For many types of DM numbers of particles of different species in the equilibrium distribution are not independent \cite{Jahnke2007}.
Another example are kinetic rates in chemical reaction networks containing cycles of reversible reactions which
are modelled by stochastic or deterministic mass action kinetics. One often requires that the
 product of reaction rates in one direction of such cycle is equal to the product of rates in the reverse direction.
This is known as detailed balance or Wegsheider's \cite{Ederer_Gilles_2007} condition and
can be intuitively explained by \mbox{time-reversal} symmetry of chemical systems
containing such reaction cycles and being in thermodynamic equilibrium \cite{Onsager_1931}. However, for the purpose of efficient computation
of variance-based sensitivity indices and for some interpretations of these indices to hold we need the parameters
considered for sensitivity analysis to be independent.
For kinetic rates this can be achieved for instance by treating some of them as independent and
using Wegscheider's conditions to compute the remaining ones \cite{Yang_Bruno_Hlavacek_Pearson_2006, Cristaldi_2011} or by changing the
parametrization so that in the new one the thermodynamic constraints are automatically observed \cite{Plested_2004, Zhang2009, Ederer_Gilles_2007}.

\section{Stochastic model with random parameters}\label{genParSec}
Below we define DM with random parameters in a way typical of Bayesian statistics
(see \cite{borovkov1999mathematical} Section 20 and
 Definition \ref{defMu} of conditional distribution in Appendix \ref{appMath}).
\begin{defin}\label{DMdef}
We say that the pair $M = (Y,P)$ consisting of a process $Y$ and random vector $P$ is DM with distribution of parameters $\mu_P $ and corresponding
to chemical reaction network $RN$, if
$P \sim \mu_P$, $Y$ is a right-continuous nonexplosive process and
$\mu_{DMCP}(RN(k), c)$ is conditional distribution of $Y$ given $P=(c,k)$.
$Y$ is called the process and $P$ the parameters of $M$.
\end{defin}
One can construct DM with distribution of parameters $\mu_P$ and corresponding to a reaction
network RN  by setting, for some $P = (C, K) \sim \mu_P$ and independent of artificial noise variable $R$ used by one of constructions
of DMCP from Section \ref{secConstr}  $c := C$ and $k:= K$ at the beginning of this construction and then proceeding with it, given
that such RN and $c,k$ always admit DMCP.
Using function $h$ (\ref{DMCPFun}) given by the construction of DMCP the process of DM we just defined can be written as
\begin{eqnarray}\label{formP}
Y = h(P,R).
\end{eqnarray}
The fact that $Y$ conforms to definition of the process of DM is consequence of (\ref{equivMarkov}) and Theorem
\ref{aveSecFin} in Appendix \ref{appMath}.
Analogously to what we did in case of process of DM we can define random parameters versions $Y$ of other types of stochastic
processes with constant parameters and provide their constructions in form of a function of independent parameters $P$ and artificial
noise $R$.  
Similarly as in Definition \ref{DMdef} of DM  we consider pairs $(Y,P)$ with the same distribution of
$P$ and conditional distribution of $Y$ given $P$ to be just different representations of the same model.

\section{Parameters of conditional distribution}\label{secCondDistr}
By observables of a process $Y$ we mean its functions $g(Y)$ which are \mbox{real-valued} random variables or
random vectors.
An observable could be for instance the number of particles of certain species at some moment of time or its maximum number over some time period.
In contrast to deterministic models with random parameters, in stochastic ones one cannot speak of a single value of the output given the parameters,
but rather of its conditional distribution given the parameters and parameters of this distribution
 like conditional expectation.
Conditional expectation of a random variable $Z \in L^1(\PR)$ (see Appendix \ref{appMath} for definition of $L^p(\PR)$
for certain probability measure $\PR$ and properties of
conditional expectation) given another variable $X$, denoted by
$\E(Z|X)$,
is formalization of the notion of the mean of $Z$ given $X$ and is a certain function of $X$.
Let us now denote $L^p_n(\PR)$ or when $\PR$ is implicitly assumed shortly $L^p_n$, to be the space of random vectors
$X= (X_i)_{i=1}^n$, such that $X_i \in L^p(\PR)$, for $i \in {I_n}$.
For $n$ bins given by numbers  $(-\infty = a_1 < a_2 < \ldots <a_{n+1} = \infty)$ histogram function $\hist$ is defined as
\begin{equation}
\hist(x): = (\I_{[a_{i},a_{i+1})}(x))_{i=1}^{n}.
\end{equation}
An example of vector-valued observable is a (single-sample) histogram
 $\hist(Z)$ corresponding to a real-valued random variable $Z$.
Note that $\hist(Z) \in L^p_n$ for every $p$ natural positive.
For random vectors $Z = (Z_i)_{i=1}^n \in L^1_n(\PR)$ and $X$ we define conditional expectation
of $Z$ given $X$ as
\begin{equation}
\E(Z|X) = (\E(Z_i|X))_{i=1}^n.
\end{equation}
Conditional histogram of $Z$ given some random variable $X$ is defined
as $\E(\hist(Z)|X)$ and mean histogram as $\E(\hist(Z))$.
For a vector $X = (X_1,\ldots, X_N)$ and any $J \subset I_N$ let
$X_J = (X_i)_{i \in J}.$
It is a well-known fact that for $Z \in L^1_n(\PR)$ for any $n$ natural positive and $J \subset K \subset I_n$ we have the following
iterated expectation property \cite{Durrett}
\begin{eqnarray}\label{doubleCond}
\E(\E(Z|X_K)|X_J) = \E(Z|X_J),
\end{eqnarray}
where by $\E(Z|X_{\emptyset})$ we mean $\E(Z)$.
%
For constructions of stochastic processes used in computer simulations, which are of form (\ref{formP}) observable of the
process can also be written as a function of parameters $P$ and the noise term $R$
\begin{eqnarray}\label{obsForm}
f(P,R) := g(h(P,R)).
\end{eqnarray}
From Theorem \ref{indepCond} in Appendix \ref{appHilb} we receive that conditional expectation of such observable can be written in the following 
intuitive form
\begin{eqnarray}\label{aveObs}
\tilde{f}(P) := \E(f(P,R)|P) = (\E(f(p,R)))_{p=P}.
\end{eqnarray}

\section{\label{appHilb}Hilbert spaces}
We now introduce some definitions and facts from Hilbert space theory, which are used in the following sections
(see references  \cite{rudin1970}  and \cite{KolmogorovFomin60} for proofs and more details).
Hilbert space is a linear space $H$, for which there exists metric $d$ induced by a norm $||\cdot||$, which is
induced by certain scalar product $(,)$ 
\begin{equation}
d(x,y) := ||x-y|| := \sqrt{(x -y,x-y)},
\end{equation}
 such that $(H,d)$ is complete metric space.
Examples of Hilbert spaces are $L^2(\mu)$ for different measures $\mu$, with scalar product given by
\begin{equation}
(f,g) = \int \! fg \, \mathrm{d}\mu.
\end{equation}
For linear subspaces $W_1,\ldots,W_n$ of certain linear space their sum is denoted and defined as follows
\begin{equation}
\sum_{i=1}^n W_i: = \{\sum_{i=1}^{n}w_i:\ \forall i \in I_n\quad w_i \in W_i\}.
\end{equation}
\begin{defin}\label{defHilb}
 Hilbert space $H$ is direct sum of its linear subspaces $H_1,\ldots, H_n$, which we denote
\begin{equation}
H = \bigoplus_{i=1}^n H_i = H_1 \oplus \ldots \oplus H_n
\end{equation}
 if the following conditions are fulfilled.
\begin{enumerate}
  \item Subspaces $H_1, \ldots, H_n$  are closed.
  \item 
\begin{equation}
H = \sum_{i=1}^n H_i.
\end{equation}
  \item These subspaces are mutually orthogonal, that is for every $i,j \in I_n$, $i \neq j$ for every $v_i \in H_i$ and $v_j \in H_j$
\begin{equation}
 (v_i,v_j) = 0.
\end{equation}
\end{enumerate}
\end{defin}
It turns out that elements $v_i \in H_i$ for $i \in I_n$ such that
\begin{equation}
v = \sum_{i=1}^n v_i
\end{equation}
are uniquely determined.
Since for every $J \subset S_n$ the subspace
\begin{equation}
H_J =  \sum_{i \in J} H_i
\end{equation}
can be proved to be closed, thus it is Hilbert space for which we further have
$H_J = \bigoplus_{i\in J} H_i$.
For any partition $\{J\cup K\}$ of $I_n$ it holds
\begin{equation}\label{sumDiv}
H = H_J\oplus H_K.
\end{equation}
We define direct product of $n$ Hilbert spaces $(H_i)_{i=1}^n$ with respective scalar products
$((\cdot,\cdot)_i)_{i=1}^n$ to be the the Cartesian product space $H_1\times\ldots \times H_n$ with scalar product defined as
\begin{equation}
(v,w) = \sum_{i=1}^n (v_i,w_i)_i.
\end{equation}
It can easily be proved to be complete, thus it is Hilbert space.
If $M$ is any closed subspace of $H$ then $M^{\perp} = \{v \in H: \forall w \in M\quad v \perp w\}$
is the unique subspace of $H$ for which it holds
\begin{equation}\label{mDirect}
H = M \oplus M^{\perp}.
\end{equation}
For every $v \in H$ the uniqueness of decomposition
\begin{equation}
 v = v_M + v_{M^{\perp}},
\end{equation}
where $v_M \in M$ and $v_{M^{\perp}}\in M^{\perp}$, allows to define a linear function $P_M$ from $H$ onto $M$,
such that $P_M(v) = v_M$.  $P_M$ is called orthogonal projection of $H$ onto $M$. $v_M$
is the unique element of $M$ minimizing distance from $v$, that is
\begin{equation}\label{infProp}
d(v, v_M) = \inf_{w \in M} d(v,  w)
\end{equation}
and it holds
\begin{equation}\label{distP}
d(v, v_M)^2 = ||v||^2 - ||v_M||^2.
\end{equation}

\section{Conditional expectation as orthogonal projection and generalizations of variance}\label{secOrthog}
$L^2(\PR)$ is Hilbert space  with scalar product $(,)$ defined as
\begin{equation}\label{1scalar}
(X,Y): = \E(XY).
\end{equation}
We denote the norm it induces $||\cdot ||$ and the metric $d$.
For some $n$ natural positive let $<,>_n$ be any scalar product on $\R^n$. Let $(a_{ij})_{i,j \in I_n}$
be real numbers such that for every $x,y \in \R^n$ we have
\begin{equation}
<x,y>_n = \sum_{i,j\in I_n} a_{ij} x_iy_j.
\end{equation}
For instance for the standard scalar product we have $a_{ij} = \delta_{ij}$, where $\delta_{ij}$ is Kronecker delta.
We denote the norm induced by $<,>_n$ as $|\cdot|_n$ and the distance it induces $\dist_n$.
We define Hilbert space on $L^2_n$ for any $n$ natural positive by equipping it with scalar product $(,)_n$ defined
for $X,Y \in L^2_n$ as
\begin{equation}\label{scalarFun}
(X,Y)_n: = \E(<X,Y>_n) = \sum_{i,j\in I_n}a_{ij}(X_i,Y_j). 
\end{equation}
We denote the norm it induces by $||\cdot ||_n$ and the distance $d_n$.
We say that 2 norms $|\cdot|_1, |\cdot|_2$ on the linear space $A$ are equivalent, if there exist $\alpha$ and $\beta$ real positive such that
\begin{equation}
\forall x \in A \quad |x|_1 \leq \alpha|x|_2 \leq \beta|x|_1.
\end{equation}
The completeness of $L^2_n$ with norm induced by any above defined scalar product $(,)_n$ is a consequence of
the fact that  for $<,>_n$ equal to standard scalar product the defined space becomes an $n$-fold direct sum of $L^2$,
 which is complete (see Section \ref{appHilb})
and the well-known fact that all norms in finite dimensional spaces like $\R^n$ are equivalent and 
from (\ref{scalarFun}) so are different $||\cdot ||_n$.

Let us denote $L^2_{n,X}$ ($L^2_{X}$) to be the subspace of $L^2_n$ ($L^2(\PR)$) consisting of all its elements being
certain functions of random variable $X$. This is a closed subspace.
\begin{theorem}\label{pytag}
If $Z \in (L^2_n)$ and $X$ is a random variable, then $\E(Z|X)$ is orthogonal projection of $Z$ onto $L^2_{n,X}$.
\end{theorem}
\begin{proof}
We have $\E(Z|X) \in L^2_{n,X}$. Furthermore, for any $f(X) = (f_j(X))_{j=1}^n \in L^2_{n,X}$ we have
\begin{equation}
(Z - \E(Z|X),f(X))_n = \sum_{i,j\in I_n} a_{ij}(Z_i - \E(Z_i|X),f_j(X)) = 0,
\end{equation}
since for all $i \in I_n$ it holds $Z_i - \E(Z_i|X) \in (L^2_{X})^{\perp}$,
because $\E(\cdot|X)$ is orthogonal projection from $L^2(\PR)$ onto $L^2_{X}$ (see \cite{Durrett}, Sec. 4.1 Theorem 1.4).
 We thus have $Z - \E(Z|X) \in (L^2_{n,X})^{\perp}$.
\end{proof}
As othogonal projection, $\E(Z|X)$ is the best approximation of $Z$ among all functions of $X$ in $L^2_n$ and the error of this approximation fulfills
\begin{equation}\label{condError}
d_n(Z, \E(Z|X))^2 =  ||Z||_n^2 - ||\E(Z|X)||_n^2.
\end{equation}
For random variable $Z \in L^2(\PR)$ by its variance we mean
 \begin{equation}\label{varDef}
 \Var(Z) = \E(Z - \E Z)^2 = \E (Z^2) - \E^2( Z),
 \end{equation}
 while by conditional variance of $Z$ given $X$
 \begin{equation} \label{defCondVar}
\Var(\E(Z|X)): = \E (Z^2|X) - \E^2(Z|X) = \E((Z - \E(Z|X))^2|X).
 \end{equation}
%

We generalize variance to random vectors $Z \in L^2_n$ as follows
\begin{equation}\label{varGen}
\Var(Z) := d_n(Z, \E(Z))^2 = ||Z||^2_n - ||\E(Z)||^2_n
\end{equation}
and conditional variance of $Z$ given $X$ as
\begin{equation}\label{condGen}
\Var(Z|X): =  \E(\dist_n(Z, \E(Z|X))^2|X).
\end{equation}
Using iterated expectation property (\ref{doubleCond}) we rewrite (\ref{condError}) to receive generalized version of a well-known formula 
\begin{equation} \label{aveVarError}
\Var(Z) = \E(\Var(Z|X)) + \Var(\E(Z|X)).
\end{equation}
One can further generalize variance and conditional variance to the case of random vectors
 by using  metrics $D_n$ on $\R^n$, which are not induced by scalar products and defining variance as $\E(D_n(Z, \E(Z))^2)$ and
conditional variance as in (\ref{condGen}) with $\dist_n$ replaced with $D_n$.
For instance Degasperi et. al.
 ( \cite{Degasperi2008} and e-mail communication with Mr Degasperi) apply $k$-dimensional Manhattan distance for some $k$ natural
\begin{equation}\label{ManhMetric}
DM_k(X,Y) = \sum_{i=1}^{k} |X_i - Y_i|
\end{equation}
to compute such defined variances for conditional histogram $Z = \E(\hist(f(P,R))|P)$ of some observable
 $f(P,R)$ of the process of DM and variance of $\E(Z|X)$ for $X = P_J$, where $P_J$ is certain subvector of $P$.
Such variances are called main-sensitivity indices of $Z$ with respect to $P$ and  $P_J$ respectively and we
discuss them in more detail in further Sections.
Unfortunately, for variances and conditional variances defined using metric (\ref{ManhMetric}),
 formula (\ref{aveVarError}), which is crucial for some applications of variance-based sensitivity indices, in general does not hold.
  For instance let us consider histogram function $h$ with bins given by
  $(a_1=-\infty,\ a_4 = \infty \text{ and } a_{i+1} = -3 + i*2\text{ for } i \in I_2)$ and two independent
   random variables $\epsilon_1, \epsilon_2$ with distribution $\PR(\epsilon_i = 1) = \PR(\epsilon_i = -1) = \frac{1}{2}$, for
  $i \in I_2$.  For $Z = \epsilon_1 + \epsilon_2$ we have $\Var(h(Z)) = \frac{13}{8}$, $\Var(\E (h(Z)|\epsilon_1)) = \frac{1}{4}$ and
  \begin{equation}
  \E (D_3(h(Z),\E(h(Z)|\epsilon_1)^2) = 1,
  \end{equation}
  thus the counterpart of expression (\ref{aveVarError}) does not hold.

\chapter{Variance-based sensitivity analysis}\label{varBased}

\section{ANOVA decomposition}\label{secANOVA}
In this whole Section $X=(X_1, \ldots, X_N)$ is a random vector with independent coordinates, $f$ is a function such that $f(X)=(f_i(X))_{i=1}^n 
\in L^2_n$ for some $n$ natural positive. We denote $I: = I_n$. 
For $J \subset I$ vector $X_J$ is defined as in Section \ref{secCondDistr}.  
For $J \neq \emptyset$ $L^2_{X_J}$ is defined as in Section \ref{secOrthog} and $L^2_{X_{\emptyset}}$ denotes the set of all real constants. 
We further denote $X_{\sim i}$ to be the sub vector of $X$ with all its coordinates except for the $i$-th. 
For each $J \subset I$ let us denote $L^2_{n,J}$ to be the subspace of  $L^2_{n,X_J}$ consisting of variables $Z = f(X_J)$ such that 
for every $i \in J$ we have 
\begin{equation}\label{zeroInt} 
\E(Z|X_{\sim i}) = \int \! f(X_{J \setminus \{i\}}, x_i) \, \mu_i(\mathrm{d} x_i) = 0, 
\end{equation} 
where expression in the middle is a convenient notation for integrating only the $i$-th variable over its distribution and 
the first equality is a consequence of Theorem \ref{indepCond} from Appendix \ref{appMath}. 
Note that $L^2_{n,\emptyset}$ denotes the subspace of constant vectors. From (\ref{zeroInt}) 
and iterated expectation property (\ref{doubleCond}) it follows that elements $Z_J \in L^2_{n,J}$ for $J \neq \emptyset$ fulfill 
\begin{equation}
\E(Z_J) = 0.
\end{equation} 
We now introduce generalization of well-known ANOVA decomposition to the case of elements of $L^2_{n,X}$. 
To our knowledge ANOVA decomposition for real-valued variables appeared for the first time in  \cite{Efron1981}.  
See \cite{Kuo_2010} and \cite{Archer1997} for different formulations and alternative proofs of this decomposition 
for the special case of real-valued functions and for reviews of its history. 
\begin{theorem}
For Hilbert space $L^2_{n,X}$ with certain scalar product $(,)_n$ as discussed in Section \ref{secOrthog}, we have 
\begin{equation}
L^2_{n,X} = \bigoplus_{J \subset I} L^2_{n,J}. 
\end{equation}
\end{theorem}
\begin{proof}
For every $J \subset I$ set $L^2_{n,J}$ is closed in $L^2_{n,X}$, since it is intersection 
of closed set $L^2_{n,X_J}$ and $\E^{-1}(\cdot|X_{\sim i})(0)$ for $i \in J$, 
which are closed due to conditional expectations being continuous as any orthogonal projections.
We need to prove that for every $f(X) \in L^2_{n,X}$ there exist $f_J(X_J) \in L^2_{n,J}$ for $J \subset I$, such that we have 
\begin{equation}\label{anovaDec} 
f(X) = \sum_{J \subset I}f_J(X_J). 
\end{equation} 
Notice that for $K,J \subset I$ such that $\exists i \in K \cap (I \setminus J)$ from $f_K(X_K) \in L^2_{n,K}$ it follows 
\begin{equation}\label{condDisap}
\E(f_K(X_K)|X_J) = \E(\int f_K(X_{K \setminus \{i\}},x_i)\,  \mu_i(dx_i)|X_J) = 0. 
\end{equation}
Applying conditional expectation $\E(\cdot|X_{_J})$ to both sides of (\ref{anovaDec}) for $J \subset I$ and using (\ref{condDisap}) we receive 
set of formulas 
\begin{equation}\label{recufJ} 
\{\E(f(X)|X_J) = \sum_{K \subset J} f_K(X_K)\}_{J \in I},
\end{equation}
which uniquely determine every $f_J(X_J)$ on the right hand side (rhs)
 of (\ref{anovaDec}) (full proof would follow by induction over $|J|$ - the size of $J$). 
The fact that such defined $f_K(X_K)$ are in respective spaces $L^2_{K,n}$ follows by induction. For $f_\emptyset = \E(f(X))$ it is obvious. 
Let us assume for certain $k<n$ it holds for all $\{J\subset I :|J|\leq k\}$. We prove it for $\{ J:|J| = k+1 \}$ as follows. For $i \in J$ from formulas 
(\ref{recufJ}) we have 
\begin{equation}\label{indStep}
\E (f_J(X_J)|X_{\sim i}) = \E(f(X)|X_{J \setminus \{i\}}) - \sum_{K \subset (J \setminus \{i\})} f_K(X_K),
\end{equation}
since 
\begin{equation}
\E(\int\ f(X_{J \setminus \{i\}}, x_i) \mu_i(dx_i)| X_J) = \E(f(X)|X_{J \setminus \{i\}})
\end{equation}
and by inductive hypothesis for $K \subsetneq J,$ 
$i \in K$ $f_K(X_K)$ become $0$ when applying to them  $\E(\cdot|X_{\sim i})$, while $f_K(X_K)$ for $i \notin K$ remain unchanged. 
From (\ref{recufJ}) the rhs of (\ref{indStep}) is equal to zero vector. 
Finally, for $J$, $K \subset I$, $J \neq K$ we need to prove that $f_J(X_J)$ and $f_K(X_K)$ are orthogonal.
Without loss of generality assuming that 
there exists certain $i \in J\setminus K$, for every $l,m \in I$ we have 
\begin{equation}
\E (f_{J,l}(X_J)f_{K,m}(X_K)) = \E (\int \! f_{J,l}(X_{J \setminus \{i\}}, x_i) \, \mu_i(\mathrm{d} x_i)f_{K,m}(X_K))= 0,
\end{equation}
so from expression (\ref{scalarFun}) we receive $(f_J(X_J), f_K(X_K))_n = 0$.
\end{proof} 
Vector $(f_J(X_J))_{J \subset I}$ as in the above theorem  
is called ANOVA decomposition of $f(X)$. 
Denoting for $J \subset I$ 
\begin{equation}
V_J := \Var(f_J(X_J))
\end{equation}
and using (\ref{anovaDec}) and orthogonality of elements of ANOVA decomposition we get  
\begin{equation}\label{sumvar}
\Var(f(X))  =  \sum_{K \subset J} V_K.
\end{equation} 
We call $(V_K)_{K \subset I}$ the ANOVA decomposition of variance of $f(X)$. For convenience instead of writing $V_{\{i,j \ldots, k\}}$  
we write simply $V_{i,j \ldots, k}$. 
For a family $\mathcal{S}$ of subsets of $I$ such that 
\begin{equation}
\forall A \in \mathcal{S}\quad \forall B \subset A\quad B \in \mathcal{S}
\end{equation}
 we have 
\begin{equation}
\sum_{J \in \mathcal{S}} L^2_{n,X_J} = \bigoplus_{J \in \mathcal{S}} L^2_{n,J}.
\end{equation} 
 From expressions (\ref{mDirect}) and (\ref{infProp}) from Section 
\ref{appHilb} it thus follows that $\sum_{J \in \mathcal{S}} f_J$ is 
the best approximation of $f(X)$ among 
linear combinations of functions of one of sub vector from the set $\{X_J\}_{J \in \mathcal{S}}$. From (\ref{sumvar}) and (\ref{distP}) error of
 this approximation is equal to 
\begin{equation}\label{sumVJ} 
\sum_{J \subset I, J \notin \mathcal{S}} V_J. 
\end{equation}
For fixed $J \subset I$ substracting expression (\ref{sumVJ}) for  $\mathcal{S} = \{K: K\subsetneq J\}$ 
from this expression for $\mathcal{S}$ equal to all subsets of $J$, we receive $V_J$. Thus $V_J$ can be interpreted as the difference of errors 
of the best approximation of $f(X)$ using linear combinations of functions of proper sub vectors of $X_J$ and of the whole vector $X_J$. 
This is to our knowledge new interpretation of $V_J$, which 
has been called interaction index between variables with indices in $J$ in the literature \cite{Saltelli2005}. 

\section{Variance-based sensitivity indices}\label{vbsi}
For some $m$ natural positive let $Z \in L^2_m$ and $X = (X_i)_{i=1}^n$ be a random vector.  We denote $I = I_n$. $X_J$ 
for some $J\subset I$  is defined as in the previous Section. We assume that $\Var(Z) = D > 0$. 
One useful sensitivity index describing dependence of $Z$ on $X_J$ is 
 variance of conditional expectation of $Z$ given $X_J$, which is known as the main sensitivity index of $Z$ 
 with respect to $X_J$ \cite{Archer1997}
\begin{eqnarray}\label{VXJ} 
V_{X_J} := \Var(\E(Z|X_J)). 
\end{eqnarray} 
We call its normalized version 
\begin{equation}\label{SXJ} 
S_{X_J} = \frac{V_{X_J}}{D} 
\end{equation} 
 Sobol's main sensitivity index \cite{Saltelli2005}. 
From (\ref{aveVarError}) $D - V_{X_J}$ is equal to the error of the best approximation of $Z$ in $L^2_{m,X_J}$, 
in particular when $S_{X_J} = 1$ we receive that $Z$ is a function of $X_J$. 
Let us now assume that  $Z= f(X)$  for certain function $f$ and random vector $X = (X_i)_{i=1}^n$. 
Denoting $\sim J := I \setminus J$, we define the total sensitivity index $f(X)$ with respect to $X_J$ \cite{Saltelli2005}
\begin{eqnarray}\label{VXJtot}
V_{X_J}^{tot} := D - V_{X_{\sim J}}, 
\end{eqnarray}
where $\sim J = I \setminus J$. Its normalized version, is called Sobol's total sensitivity index 
\begin{equation}\label{SXJtot}
S_{X_J}^{tot} = \frac{V_{X_J}^{tot}}{D}.
\end{equation}
$V_{X_J}^{tot}$ is equal to the error of the best approximation of $f(X)$ among functions of $X_{\sim J}$ from $L^2_m$. In
particular
\begin{eqnarray}\label{totExtr}
S_{X_J}^{tot} = 0 \Leftrightarrow f(X) = \E(f(X)|X_{\sim J}).
\end{eqnarray}
Let us now assume that $X_J$ and $X_{\sim J}$ are independent. 
Since $\E(f(X)|X_{\sim J})$ is function of $X_{\sim J }$ the rhs of (\ref{totExtr}) implies that $f(X)$ is independent of $X_J$. 
 By treating $f(X)$ as a function of two arguments $X_J$ and $X_{\sim J}$ we get its variance decomposition sum analogical to (\ref{sumvar})
 \begin{eqnarray}
 D = V_{X_J} + V_{X_J, X_{\sim J}} +  V_{X_{\sim J}}.
 \end{eqnarray}
 where $V_{X_J, X_{\sim J}}$ is the interaction index of $X_J$ and $X_{\sim J}$. 
 From (\ref{VXJtot}) we get that 
 \begin{eqnarray}\label{totDec}
 V_{X_J}^{tot} = V_{X_J} + V_{X_J,X_{\sim J}}.
 \end{eqnarray}
If all coordinates of $X$ are independent, using (\ref{sumvar}) we rewrite (\ref{VXJtot}) as 
\begin{equation}\label{sth}
 V_{X_J}^{tot} = \sum_{K \subset I: K \cap J \neq \emptyset} V_K,
\end{equation}
which is the sum of all interaction terms involving indices from $J$. 
Equations (\ref{totDec}) and (\ref{sth}) provide some intuition for the name total effect of $X_J$ on $f(X)$ for $V_{X_J}^{tot}$ and prove that
\begin{eqnarray}\label{ineqS}
 0 \leq S_{X_J} \leq S_{X_J}^{tot} \leq 1.
\end{eqnarray}
For $X_J$ and $X_{\sim J}$ not independent
neither inequalities (\ref{ineqS}) nor the fact that $S_{X_J}^{tot} = 0$ implies independence of $X_J$ and $Z$ are true, for instance 
for $X_J = X_{\sim J}$  we get $S_{X_J} = 1$, $S_{X_J}^{tot} = 0$ and 
$f(X)$ is not independent of $X_J$. 

\section{Sensitivity indices for observables of DM}\label{sensAveSec}
Let us consider certain observable $g(Y) \in L^2_n$ of a process $Y$ of a DM with parameters $P$.
We define main sensitivity index of $g(Y)$ given some sub vector of $P_J$ as in the previous Section. Its value $V_{P_J}$
is determined by the distribution of $\E(g(Y)|P_J)$, which is determined by distribution of  $P$ and conditional
distribution of $Y$ given $P$, and hence by Definition \ref{DMdef} of DM (see Appendix \ref{appMath}).
 Introducing a representation $f(P,R)$ (\ref{obsForm}) of the observable related to a certain construction of the process of this DM,
 we can consider some further sensitivity indices. For instance
\begin{equation}\label{VRTot}
V_{R}^{tot} := D - V_P,
\end{equation}
whose value, by inspection of rhs of (\ref{VRTot}) is also determined by definition of DM.
 However, the values of indices of $f(P,R)$ like $V_R = \Var(\E(f(P,R)|R))$ or
$V_{P_J}^{tot} = D - V_{(P_{\sim J},R)}$ are not determined by Definition \ref{DMdef} and
can be different for different constructions of DM used to define $f(P,R)$.
Let us consider the mean observable $g(Y)$ given $P$
\begin{equation}
\tilde{f}(P) := \E(g(Y)|P).
\end{equation}
Thanks to iterated expectation property ($\ref{doubleCond}$) we have $\E(\tilde{f}(P)|P_J) = \E(g(Y)|P_J)$ and therefore the main
sensitivity indices of $\tilde{f}(P)$ and  $g(Y)$ with respect to $P_J$ coincide
\begin{eqnarray}\label{tVCJ}
\tilde{V}_{P_J}: = \Var(\E (\tilde{f}(P)|P_J)) = V_{P_J},
\end{eqnarray}
while for total sensitivity indices we have
\begin{eqnarray}
\tilde{V}_{P_J}^{tot}: = \tilde{D} - \tilde{V}_{P_{\sim J}} = V_{P} - V_{P_{\sim J}}.
\end{eqnarray}
Defining $\tilde{D} := \Var(\tilde{f}(P)) = \tilde{V}_{P}$ we also have
following expressions for Sobol's sensitivity indices of $\tilde{f}(P)$
\begin{equation}\label{SCJ}
\tilde{S}_{P_J} := \frac{\tilde{V}_{P_J}}{\tilde{D}} = \frac{V_{P_J}}{V_P},
\end{equation}
\begin{equation}\label{SCJtot}
\tilde{S}_{P_J}^{tot}: = \frac{\tilde{V}_{P_J}^{tot}}{\tilde{D}} = \frac{V_{P} - V_{P_{\sim J}}}{V_P}.
\end{equation}

\section{Measures of dispersion for uncertain models}\label{secDisp}
From (\ref{aveVarError})
average conditional variance of $Z \in L^2_n$ given some variable $X$ can be expressed using main sensitivity index as follows 
\begin{equation}\label{aveCondVar}
\E(\Var(Z|X)) = D - V_{X}.
\end{equation}
Note that the last expression is equal to $V_R^{tot}$ for an observable $Z= f(P,R)$ corresponding to certain construction of DM and $X = P$.
Since variance is a measure of dispersion of distribution of model output with known parameters,
 average conditional variance given the epistemic parameters could be used to measure and
compare dispersions of models with uncertain parameters. For chemical models with constant parameters $p$
also other measures of dispersion of process observables $Z_p$ satisfying $\E(Z_p) > 0$  have been used, like
coefficient of variation
\begin{equation}
CV = \frac{\sqrt{\Var(Z_p)}}{\E (Z_p)},
\end{equation}
which is a dimensionless quantity, or Fano factor \cite{Fano47, Thattai_2001_intrinsic}
\begin{equation}
FF = \frac{\Var(Z_p)}{\E (Z_p)}.
\end{equation}
For variables with Poisson distribution variance is equal to mean and hence FF reveals whether $Z_p$ has greater variance than a Poisson variable with
the same mean. For outputs $Z$ of models with uncertain parameters one could take an average of a conditional
FF $\E \frac{\sqrt{\Var(Z|X)}}{\E Z|X}$, and similarly for CV. Instead we propose the following generalizations of
conditional variance GCV and Fano factor GFF to the random parameters case
\begin{equation}\label{CVGen}
\text{GCV} =  \frac{\sqrt{\E\Var(Z|X)}}{\E(Z)}
\end{equation}
 and
\begin{equation}\label{FFGen}
\text{GFF} = \frac{\E(\Var(Z|X))}{\E(Z)},
\end{equation}
since they can be expressed using variance based sensitivity indices as in (\ref{aveCondVar}) and hence 
are amenable for computation using our methods.
We call them generalizations, since they coincide with definitions for models with constant parameters when the distribution of epistemic
 parameters is one-point.

\section{Average variance reduction}\label{varRedSec}
We assume $Z$ is some model output, like an observable $g(Y)$ of process of DM or its mean, conditional histogram or some conditional moment
given the parameters. For random vector $X$ denoting model parameters we define its subvector $X_J$ as usual.
We rewrite expression (\ref{aveVarError}) as follows
\begin{equation}\label{varDecr}
\begin{split}
S_{X_J} &=   \frac{\E( D - \Var(Z|X_J))}{D}.
\end{split}
\end{equation}
The rhs of (\ref{varDecr}) is the normalized average difference of variance of $Z$ and its conditional variance given $X_J$.
Thus if $X_J$ are epistemic parameters, $S_{X_J}$ tells by what fraction
on average the variance of the output is reduced 
if we get to know their exact values. 
Let us assume that we can conduct an experiment measuring completely precisely one epistemic parameter, which is of course an idealisation.
If we want to achieve on average the highest reduction of the variance of the ouput, which can also be thought of as maximal reduction
 of the uncertainty or improvement of the precision of model predictions, we should measure the parameter with highest main sensitivity
index ${V}_{X_i}$.
This approach to using variance-based sensitivity indices is known as factor prioritization setting \cite{Saltelli2008}.

\section{Parameter fixing}
Let us consider a function $f(X)\in L^2_m$ of a random vector $X$, whose sub vectors $X_J$ and $X_{\sim J}$ are independent.
As shown for the case of $m=1$ in \cite{Sobol2007}
 $V_{X_J}^{tot}$ is related to the average error made when fixing variable $X_J$, in the sense we decribe and prove 
for arbitrary $m$ in this Section. 
Using certain $d_m$ and $dist_m$ as in Section \ref{secOrthog}
we define the square error of approximation of $f(X)$ when fixing $X_J$ to value $z\in\R^m$ as follows
\begin{equation}
 \Delta(z) = (d_m(f(z, {X}_{\sim J}),f(X)))^2 = \E((\dist_m(f(z, {X}_{\sim J}),f(X)))^2). 
\end{equation}
We further need the following Theorem.
\begin{theorem}\label{thCond}
For $X = (X_1,X_2)$, $Y_2 \sim X_2$ and independent of $X$ and $g(X),h(X)\in L^2(\PR)$ it holds
\begin{equation}
\E(g(X)h(X_1,Y_2)) = \E(\E(g(X)|X_1)\E(h(X)|X_1)).
\end{equation}
In particular if $g(X) = h(X)$ we receive a well-known fact \cite{Saltelli_2002} that
\begin{equation}
\E(g(X)g(X_1,Y_2)) = \E((\E(g(X)|X_1))^2)
\end{equation}
and the fact that
\begin{equation}
\Cov(g(X),g(X_1,Y_2)) = \Var(\E(g(X)|X_1)).
\end{equation}
\end{theorem}
\begin{proof}
\begin{equation}
\begin{split}
\E(g(X)h(X_1,Y_2)) &= \E(\E(g(X)h(X_1,Y_2)|X_1)) \\
&=  \E((\E(g(x_1, X_2)h(x_1,Y_2)))_{x_1 = X_1}) \\
 & = \E((\E(g(x_1, X_2)))_{x_1 = X_1}(\E(h(x_1,Y_2)))_{x_1 = X_1})\\
& = \E(\E(g(X)|X_1)\E(h(X)|X_1)),
\end{split}
\end{equation}
where in first equality we used iterated expectation property, in the second and last Theorem \ref{indepCond} and in the third independence
of $X_2$ and $Y_2$.
\end{proof}
Let ${Y}_{J} \sim X_{J}$ and be independent of $X$.
From the above Theorem it follows that
\begin{equation}\label{thCondVect}
\begin{split}
(f(X) ,f(Y_J,X_{\sim J}))_m &= \sum_{i,j \in I_m} a_{ij}(f_i(X),f_j(Y_J,X_{\sim J})) \\
&= \sum_{i,j \in I_m} a_{ij}(\E(f_i(X)|X_{\sim J}),\E(f_j(X)|X_{\sim J}))\\
& = ||\E (f(X)|X_{\sim J})||_m^2.
\end{split}
\end{equation}
Thus if $X_J$ is set randomly according to its distribution the mean square error of approximation of $f(X)$ is
\begin{eqnarray*}
\begin{split}
\E(\Delta(Y_J)) &= \E(\dist_m(f(Y_J,{X}_{\sim J}),f(X))^2)\\
&= ||f(X)||^2_m + ||f(Y_J, X_{\sim J})||^2_m - 2 (f(X),f(Y_{J}, X_{\sim J}))_m\\
&= 2(D - V_{X_{\sim J}}) = 2 V_{X_J}^{tot}.
\end{split}
\end{eqnarray*}
The normalized mean square error $\frac{\E(\Delta(Y_J))}{D}$ of the approximation mentioned is thus equal to $2S_{X_J}^{tot}$.
When $V_{X_J}^{tot} = 0$ then for 
$\mu_{X_J}$
almost every (a.e) $y_J$, for $\mu_X$
a. e. $x = (x_J,\ x_{\sim{J}})$  we have that  $f(x) = f(y_J,x_{\sim J})$.
Therefore, if we need to evaluate many independent copies of $f(X)$, for instance in a Monte Carlo simulation 
we can fix $X_J$ to some random value $y_J$ at the beginning
and evaluate independent copies of $f(y_J, X_{\sim J})$ instead. With probability $1$ we get the same result,
 but at smaller cost, as there is
no need to generate many independent copies of $X_{J}$. The cost of generation of random parameters of
chemical models is usually completely insignificant 
in comparison to the cost of function evaluation, but if $X_J$ represents component of artificial noise $R$ of some stochastic model like DM the
cost of its generation forms a noticeable fraction of the overall simulation cost \cite{Mauch_2011}.
The computation of $S_{X_J}^{tot}$ for the purpose of identifying and fixing insignificant parameters
 is known as factor fixing setting \cite{Saltelli2008}.

 One is often interested in fixing $X_J$ not to random value, but to
a certain one, possibly even lying outside the support of the distribution $\mu_{X_J}$, but
leading to significant reduction of computational cost of function evaluation.
For instance if $X$ represents kinetic rates of chemical reaction network, we may be interested in fixing some of them to $0$ or $1$, which
can lead to model reduction like removal of certain reactions \cite{Degenring2004}.
Fixing parameters $i$ with small values of $S_{X_i}^{tot}$ or even $S_{X_i}$ to certain value in order to simplify the model 
 can sometimes be useful heuristic leading to models retaining small approximation error from initial model or experimental data.
 For instance in  \cite{Cristaldi_2011} the sensitivity of error
of approximation of experimental data by the model with respect to kinetic parameters was computed for different models
describing the hydrogenation process of avermectin to ivermectin. It turned out that a simple model, which could be
created from more complex ones by removal of reactions whose propensities had
kinetic parameters with low values of main and total Sobol's indices retained
good fitting capability to experimental data representing different values of model parameters.

\chapter{Estimation methods}\label{MC}
\section{\label{secStatMod}Statistical models }
%
Statistical model is a triplet  $(B, \mathcal{B}, \mathcal{P})$, where $\mathcal{S} = (B, \mathcal{B})$ is a measurable space and
$\mathcal{P}$ is a family of admissible distributions. Functions from $\mathcal{S}$ are called statistics.
For a given $\mu \in \mathcal{P}$ random variable $X \sim \mu$ and its functions are called observables.
For fixed $\mu \in \mathcal{P}$ $\phi$ can be treated as a random variable on probability space $(\mathcal{S},\mu)$ and we denote its expectation as
\begin{equation}
\E_\mu \phi:= \int\! \phi\, d\mu.
\end{equation}
Let us consider certain real-valued function $G$ on $\mathcal{P}$, which is called estimand.
For instance if $\mathcal{S} = (\R, \mathcal{B}(\R))$ the estimand could be
the first moment of $\mu$
\begin{equation}
G(\mu) = \int\! x\, \mu(dx),
\end{equation}
assuming it exists for every $\mu \in \mc{P}$. 
Statistic $\phi$ is called estimator of $G$, if for every $\mu \in \mathcal{P}$, for any observable $X \sim \mu$, one may use $\phi(X)$
to approximate $G(\mu)$. As a measure of error of this approximation for given  $\mu \in \mathcal{P}$ one can use mean square error
\begin{equation}\label{meanSqrErr}
\E_{\mu} (\phi - G(\mu))^2.
\end{equation}
Value $\phi(X(\omega))$ corresponding to some random event $\omega$ is called estimate of $G(\mu)$.
Statistic $\phi$ is called unbiased estimator of $G$, if for every $\mu \in \mathcal{P}$, we have
\begin{equation}\label{muPhi}
\E_\mu(\phi) = G(\mu).
\end{equation}
For any statistic its variance given $\mu \in \mathcal{P}$ is defined as
\begin{equation}
\Var_{\mu}(\phi): = \E_\mu(\phi - \E_\mu(\phi))^2.
\end{equation}
Thanks to (\ref{muPhi}), for an unbiased estimator it is equal to its mean square error (\ref{meanSqrErr}) in approximating $G(\mu)$.
We further use following unbiased estimators defined on $\mathcal{S} = (\R^n, \mathcal{B}(\R^n))$ for some $n$ natural positive and 
and with $\mathcal{P}$ containing $n$-fold product measures $\mu^n$, for which their estimands exist. Estimator of first moment of $\mu$ 
\begin{equation}\label{phiAve}
\phi_{ave}(x): = \frac{1}{n} \sum_{i=1}^n x_i,
\end{equation}
of variance of any variable with distribution $\mu$ for $n \geq 2$
\begin{equation}\label{phiVar}
\phi_{var}(x): = \frac{1}{n-1} \sum_{i=1}^n (x_i - \phi_{ave}(x))^2
\end{equation}
and of variance of estimator $\phi_{ave}$ given $\mu$ for $n \geq 2$
\begin{equation}\label{phiAveVar}
\phi_{avevar}(x): =  \frac{\phi_{var}(x)}{n}.
\end{equation}

\section{Monte Carlo method}\label{secMC}
Let us assume that $\phi \in L^2(\mu)$, for some probabilistic measure  $\mu$.
Monte Carlo (MC) method is a procedure of computing estimates of integrals of form
\begin{equation}\label{intProb}
\lambda = \int\! \phi\, d\mu.
\end{equation}
Note that for statistical model containing only one admissible distribution $\mu$ $\phi$ is an unbiased estimator of estimand
$\lambda$. In such situation we say that $\phi$ is unbiased estimator of $\lambda$ with respect to $\mu$.
For independent random variables $(X_i)_{i=1}^n$,
 $X_i \sim \mu$ generated
for instance using random number generator, in each step of MC procedure one computes a value of
observable $W_i = \phi(X_i)$. Thus we call $\phi$ single-step (MC) estimator.
For $W = (W_i)_{i=1}^n$ and $\phi_{ave}$ as in the previous Section as final MC estimate of $\lambda$ one uses the
computed value of observable
\begin{equation}
\overline{W} := \phi_{ave}(W).
\end{equation}
Function defined as
\begin{equation}
\phi_{fin}(x) =  \phi_{ave}((\phi(x_i))_{i=1}^n)
\end{equation}
 is an unbiased estimator of $\lambda$
with respect to $\mu^n$ and we call it final (MC) estimator.
Let us denote the variance of single-step estimator as $\Var_s := \Var_\mu(\phi) =\Var(W_i)$ for any $i \in I_n$
and the variance of final estimator as $\Var_{a} := \Var_{\mu^n}(\phi_{ave}((\phi_i)_{i=1}^n))= \Var(\overline{W})$. It holds
\begin{equation}
\Var_a = \frac{\Var_s}{n}.
\end{equation}
As estimates of $\Var_{a}$ one uses the values of
\begin{equation}\label{varEst}
\widehat{\Var}_a(W): = \phi_{avevar}(W)
\end{equation}
and as estimates of standard deviation $\sigma_{a}$ of $\overline{W}$ the values of
\begin{equation}
\widehat{\sigma}_{a}(W):= \sqrt{\phi_{avevar}(W)}.
\end{equation}
From central limit theorem \cite{billingsley1979}, 
for large $n$ $\overline{W}$ should have approximately normal distribution. In particular
$P(|\overline{W} - \lambda| < k\sigma_a)$ is then approximately equal to $68\%$ for $k=1$ and $95\%$ for $k=2$.
We further report results of MC procedure using computed value of $\overline{W} \pm \widehat{\sigma}_{a}(W)$.

If we want to go with variance $\Var_a$ below threshold $\alpha$ for given $\Var_s$, we should use smallest number
$n$ of MC steps, such that
\begin{equation} \label{MCErr}
\Var_{a} = \frac{\Var_s}{n} < \alpha.
\end{equation}\label{approxEq}
We have
\begin{equation}
n \approx_{\epsilon} \frac{\Var_s}{\alpha},
\end{equation}
by which we mean that it holds
\begin{equation}
\frac{n - \frac{\Var_s}{\alpha}}{n} \leq \epsilon,
\end{equation}
for some $\epsilon$, which can be arbitrarily small for sufficiently small $\alpha$ (large $n$).
If most of the time of MC procedure is taken by computation of $W$ and
computation of $W_i$ in a single step lasts on average $\tau_s$, then the whole procedure lasts on average about
\begin{equation}\label{aveTime}
 n\tau_s \approx_{\epsilon} \frac{\tau_s \Var_s}{\alpha}.
\end{equation} 
Let us define  Monte Carlo step inefficiency constant as follows
\begin{equation}
c = \tau_{s} \Var_{s}.
\end{equation}
Let us assume that the same number $\lambda$ can be estimated in MC procedures using different functions $\phi_i \in L^2(\mu_i)$, 
having different values of respective mean duration times $\tau_{s,i}$ of single step and variances of single step 
estimators $\Var_{s,i}$ for $i$ in some set $A$. 
When our goal is to minimize the average computation time needed to go below given error $\alpha$ and $\epsilon$ in the counterpart of equality
(\ref{approxEq}) for each function $\phi_i$ is negligibly small,
then from (\ref{aveTime}) we should choose $i$ with minimum value of respective inefficiency constant $c_i$.
Let us assume that for some $n_i$ and $\tau_i$ for $i \in I_2$ denoting the number of steps and average duration of each step for 
two different Monte Carlo procedures respectively, we have approximate identity
\begin{equation}
n_1 \tau_1 \approx_{\delta} n_2\tau_2.
\end{equation}
The ratio of variances of final MC estimators $\Var_{a,j}$ of these procedures is then
 approximately the same as the ratio of their respective inefficiency constants $c_{j})$
\begin{equation}\label{varAiRatio}
\frac{\Var_{a,1}}{\Var_{a,2}} = \frac{\Var_{s,1}n_2}{\Var_{s,2}n_1} \approx_{\delta} \frac{c_{1}}{c_{2}},
\end{equation}
where $\Var_{s,i}$ for $i \in I_2$ are the variances of their respective single-step estimators. 

\section{Schemes for estimation of sensitivity indices}\label{secUnbiased}
For a given $N$ natural positive let $\mathcal{R}_N$ be the set of all pairs $(f, \mu)$ of product probability distributions
 $\mu = \mu_1 \times \ldots \times \mu_N$ and measurable functions $f$, such that $\mu$ and $f$ are defined on the same product of
measurable spaces.
\begin{defin}
For a given $N$ natural positive,
by (generalized) estimand of order $N$ we mean a real-valued function $G$ on some subset $\mc{R}_{G}$ of $\mathcal{R}_N$.
\end{defin}
As an example of such generalized estimand of order $2$
we define total sensitivity index $V_1^{tot}$ of functions of two arguments with respect to the first argument by demanding that
$\mc{R}_{V_1^{tot}} = \{(f, \mu)\in \mc{R}_2: f \in L^2(\mu), \}$
 and for any $(f, \mu)\in \mc{R}_{V_1^{tot}}$ and any $(X_1,X_2) \sim \mu$
\begin{equation}
V_1^{tot}(f)(\mu): = \E(f(X)^2) - \E(\E^2(f(X)|X_{2})).
\end{equation}
For product distribution $\mu = \mu_1 \times \ldots \times \mu_N$ and $v = (v_i)_{i=1}^N$ we define
\begin{equation}
\mu^v = \mu_1^{v_1}\times \ldots \times \mu_N^{v_N}.
\end{equation}
\begin{defin}\label{defGEst}
A (generalized) unbiased estimator $\phi$ of estimand $G$ of order $N$ on $\mathcal{R}_{G}$ with (vector of) dimensions of arguments
$v = (v_1, \ldots, v_N)$ is a function on the set
$\{f: (f, \mu) \in \mc{R}_{G}\}$, such that for any $(f,\mu) \in \mathcal{R}_{G}$
$\phi(f)$ is unbiased estimator of $G(f,\mu)$ with respect to $\mu^v$. In other
words, if $\mu = \mu_1 \times \ldots \times \mu_N$, then for any random vector $\wt{X} = ((\wt{X}_{i,j})_{i=1}^{v_j})_{j=1}^N$, whose elements
are mutually independent and fulfill $\wt{X}_{i,j} \sim \mu_j$, we have
\begin{equation}\label{phifX}
 \E(\phi(f)(\wt{X})) = G(f,\mu).
\end{equation}
\end{defin}
Using notations from Definition \ref{defGEst}, for $\phi(f) \in L^2(\mu)$ we denote
\begin{equation}
\Var_{f,\mu}(\phi) := \Var_{\mu^{v}}(\phi(f)) = \Var(\phi(f)(\wt{X})).
\end{equation}
For any sets $B_1, \ldots, B_N$ and
\begin{equation}\label{Bprod}
B = B_1 \times \ldots \times B_N
\end{equation}
 we denote
\begin{equation}
B^{v} = B_1^{v_1} \times \ldots \times B_N^{v_N}.
\end{equation}
We also denote
\begin{equation}
I_v := I_{v_1}\times \ldots \times I_{v_N}.
\end{equation}
For a point  $\wt{x} = ((\wt{x}_{i,j})_{j=1}^{v_i})_{i=1}^N \in B^{v}$ and $j= (j_i)_{i=1}^N \in I_v$  we denote
\begin{equation}
\wt{x}_j: = (\wt{x}_{i,j_i})_{i=1}^N.
\end{equation}
and for any $i \in I_N$
\begin{equation}\label{wtxi}
\wt{x}_i = (\wt{x}_{i,l})_{l=1}^{v_i}.
\end{equation}
For any $j \in I_v$ we denote by $g_j$ a function on the set of all real-valued
functions $f$ on any product sets $B$ as in (\ref{Bprod}), such that
\begin{equation}\label{gjDef}
g_{j}(f): B^{v} \longmapsto \R :  g_{j}(f)(\wt{x})  = f(\wt{x}_j).
\end{equation}
We also denote, for any finite subset $A \subset I_v$
\begin{equation}
g_A := (g_j)_{j \in A}.
\end{equation}
We define 
\begin{equation}
(A)_i := \{j_i: j \in A\},
\end{equation}
\begin{equation}\label{nAi}
n_{A,i} = \max\{k: k \in (A)_i\},
\end{equation}
and $n_A = (n_{A,i})_{i=1}^N$.
We denote $\N_{+}$ to be the set of positive natural numbers.
The concept of scheme for estimation we introduce below can be thought of as a certain general method for obtaining unbiased
estimators.
\begin{defin}
A scheme (of order N)for estimation of generalized estimand $G$ of order $N$ is a pair $(t, A)$, for a finite set $A \subset \N_{+}^N$ and a function
\begin{equation}
t:\R^{|A|} \longmapsto \R,
\end{equation}
such that
\begin{equation}\label{phiAF}
\phi(t,A) =  t(g_A)
\end{equation}
is generalized unbiased estimator of $G$ with dimensions of arguments $n_A$.
\end{defin}
Estimator $\phi(t,A)$ defined by (\ref{phiAF}) is called estimator corresponding to, or given by scheme $(t, A)$.
For example for generalized estimand $V_1^{tot}$ we introduced earlier in this Section, the scheme $(t, A)$ is defined as follows.
We take $A= \{(1,1), (2,1)\}$ and
\begin{equation}
t(x_{(1,1)}, x_{(2,1)}) = x_{(1,1)}^2  - x_{(1,1)}x_{(2,1)}.
\end{equation}
Introducing C language-like notation $g_j = g[j_1-1]\ldots[j_N-1]$, estimator corresponding to $(t,A)$ can be written as
\begin{equation}\label{V1a3tot}
\widehat{V}_{1,a3}^{tot} = g[0][0](g[0][0] - g[1][0]).
\end{equation}
The fact that this is scheme for estimation of $V_1^{tot}$ is a consequence of Theorem \ref{thCond} and the fact that
observable of this estimator corresponding to function $f$ and observable $\wt{X} = (\wt{X}_{1}[j]_{j=0}^{1}, \wt{X}_{2}[0]) \sim \mu^{n_A}$
 is
\begin{equation}\label{obsVitot}
f(\wt{X}_1[0],\wt{X}_2[0])(f(\wt{X}_1[0],\wt{X}_2[0]) - f(\wt{X}_1[1],\wt{X}_2[0])).
\end{equation}
We use formulas defining estimators of form like (\ref{V1a3tot}) to concisely
define schemes for estimation.
Scheme given by formula like (\ref{V1a3tot}) is a pair $(t, A)$, with set $A$ of indices $j$ corresponding to different $g_j$ appearing on
the rhs of this formula and $t$ acting on its arguments $(x_j)_{j \in A}$ the same way as function of $(g_j)_{j \in A}$
given by the rhs of formula like (\ref{V1a3tot}) acts on its arguments.

Note that any estimation scheme $(t,A)$ for estimation of some estimand can be used to generate estimates of its values corresponding
to some $(f, \mu)$ as follows.
One first generates the required values $\wt{X}_{j}$ for $j \in A$ and $\wt{X} \sim \mu^{n_A}$, then computes
values of $g_j(f)(\wt{X}) = f(\wt{X}_{j})$ and finally uses them to compute the value of $t$ on $g_A(f(\wt{X}))$.
We thus call 
$|A|$ the number of function evaluations used by scheme $(t,A)$.
Such computed values can be used as single-step Monte Carlo estimates. For instance scheme analogous to (\ref{V1a3tot})
can be used in Monte Carlo procedure estimating $V_R^{tot}$ for some observable $f(P,R)$ of process of DM corresponding to one of its constructions.
In our numerical experiments using different schemes, functions $f$ corresponding to different constructions of DM and distributions
$\mu$ the duration $\tau_{s}$ of a single MC step for the same $f$ and $\mu$ and using the same implementation of process
simulation algorithm on the same computer is with a good approximation
proportional to the number of function evaluations used by different schemes. Let this
proportionality constant for given $f$, $\mu$, implementation and computer be $\tau_k$.
The inefficiency constant of a single MC step using scheme $(t, A)$ can be written as
\begin{equation}\label{cIneff}
c = \Var_{f, \mu}(\phi(t, A))|A|\tau_k.
\end{equation}
Note that if we are interested in estimating variance-based sensitivity index of an observable of DM, whose value does not depend on its construction
(see discussion in Section \ref{sensAveSec}), then we can use given scheme for different functions $f$ appearing in observables
$f(P,R)$ corresponding to different constructions of DM and with $\mu = \mu_P \times \mu_R$ for $\mu_R$ corresponding to distributions of
noise variables used in these constructions. This may influence the value of $\Var_{f, \mu}(\phi)$.
Numerical results and some intuitions concerning these differences of variances for estimators using GD and RTC constructions
for different schemes  are discussed in Section \ref{secVarDiff}.
The time $\tau_k$ depends not only on the construction of DM used, but also on its computer implementation and even computer architecture,
which is discussed in more detail in Section \ref{secImpl}.
For single steps of MC procedures, whose inefficiency constants are approximately of form (\ref{cIneff}), and which use
the same functions $f$ with the same distribution $\mu$ and the same implementation on the same computer,
the ratio of their inefficiency constants is the same as of implementation-independent
inefficiency constants defined as
\begin{equation}\label{dIneff}
d_{t,A}(f,\mu) = \Var_{f,\mu}(\phi(t,A))|A|.
\end{equation}
We call (\ref{dIneff}) inefficiency constant of scheme $(t, A)$ corresponding to  $f$ and $\mu$. 
Similarly as in (\ref{varAiRatio}) one shows that
the ratio of inefficiency constants (\ref{dIneff}) of two different schemes for given $\mu$ and $f$
is equal to the ratio of variances of final MC estimators using these schemes
for the same number of function evaluations made in both MC procedures.

\section{Symmetrization of schemes}
Let $\Pi_a$ be the group of bijections, which we also call permutations of set $\N_{+}^N$.
Subgroup of $\Pi_a$ consisting of permutations of the $i$-th coordinate is defined as
\begin{equation}\label{Piai}
\Pi_{a,i}: = \{\pi\in  \Pi_a: \forall j \in \N_{+}^N \quad (\pi(j))_k = j_k \text{ for } k \neq i\}.
\end{equation}
For two subgroups $\Pi_{s_1}, \Pi_{s_2}$ of $\Pi_a$ we denote $\Pi_{s_1}\cdot \Pi_{s_2}$ to be its subgroup generated by elements
$\pi_1\pi_2$, such that $\pi_i \in \Pi_{s_i}$ for $i \in I_2$.
Let us consider subgroup $\Pi_b$ of $\Pi_a$ defined as $\Pi_b = \Pi_{a,1}\cdot \ldots \cdot\Pi_{a,N}$.
Let us consider some finite subgroup $\Pi_s$ of $\Pi_b$.
For finite set $A \subset \N_{+}^N$ we define its symmetrization with respect to $\Pi_s$ as
\begin{equation}
\Pi_s(A) = \{\pi(j): j \in A, \ \pi \in \Pi_s\}.
\end{equation}
For function $t: \R^{|A|} \longmapsto \R$
we define symmetrization of $t$ with respect to $\Pi_s$ and $A$, denoted as $S_{\Pi_s,A}(t)$ to be a function
from $\R^{|\Pi_s(A)|}$ to $\R$, such that
\begin{equation} \label{symOp}
S_{\Pi_s,A}(t)((y_j)_{j \in \Pi_s(A)}) =  \frac{1}{|\Pi_s|} \sum_{\pi \in \Pi_s} t((y_{\pi(j)})_{j \in A}).
\end{equation}
Symmetrization of scheme $(t,A)$ for estimation of $G$ with respect to $\Pi_s$ is defined as scheme
 $(S_{\Pi_s,A}(t),\Pi_s(A))$. The estimator given by  $(S_{\Pi_s,A}(t),\Pi_s(A))$ is
\begin{equation}
  S_{\Pi_s,A}(t) (g_{\Pi_s(A)}) = \frac{1}{|\Pi_s|} \sum_{\pi \in \Pi_s} t((g_{\pi(j)})_{j \in A}).
\end{equation}
Note that an observable of estimator of symmetrized scheme, corresponding to some $f$ and $\wt{X} \sim \mu^{n_{\Pi_s(A)}}$,
which can be written as
\begin{equation}\label{symObs}
\frac{1}{|\Pi_s|} \sum_{\pi \in \Pi_s} t((f(\wt{X}_\pi(j)))_{j \in A}),
\end{equation}
is a sum of random variables with the same distribution. Thus symmetrized scheme is also scheme for estimation of $G$.
We define subgroup of $\Pi_{a,i}$ (\ref{Piai}) consisting of permutations of first $k$ indices in the $i$-th coordinate as
\begin{equation}
\Pi_{i,k} = \{\pi \in \Pi_i: \forall j \in \N_{+}^N \quad (\pi(j))_i = j_i \text{ for } j_i \notin I_k\}.
\end{equation}
Symmetrization of a scheme $(t,A)$ with respect to $\Pi_{i,k}$ for $k \geq n_{A,i}$ is called symmetrization in the $i$-th argument from 
$n_{A,i}$ to $k$ dimensions,
or if $k=n_{A,i}$ simply symmetrization in the $i$-th argument.
Symmetrization with respect to $\Pi_A = \Pi_{1,n_{A,1}}\cdot \ldots \cdot \Pi_{N,n_{A,N}} $ is called symmetrization of the scheme
in all arguments.

We further need the following well-known Theorem we leave without proof.
\begin{theorem} \label{thCov}
For  $X,Y \in L^2(\PR)$ and $X \sim Y$ we have that
\begin{equation}\label{covIneq}
\Cov(X,Y) \leq \Var(X).
\end{equation}
Equality in (\ref{covIneq}) holds if and only if (iff) $X=Y$.
\end{theorem}
From theorem below it follows that estimator corresponding to symmetrized scheme has not higher variance than the one given by the original scheme.
\begin{theorem}\label{thAveVar}
If random variables $(A_i)_{i=1}^n$ from $L^2(\PR)$ have the same distribution, then
\begin{equation}\label{sumAi}
\Var(\frac{1}{n}\sum_{i=1}^{n}A_i) \leq \Var(A_1)
\end{equation}
and equality in (\ref{sumAi}) holds iff $A_i = A_j$ for all $i,j \in I_n$.
\end{theorem}
   \begin{proof}
   We have
     \begin{equation}
     \begin{split}\label{Aiineq}
     \Var(\frac{1}{n}\sum_{i=1}^{n}A_i) &= \frac{1}{n^2}\sum_{i,j\in I_n} \Cov(A_i,A_j)\\
     &\leq \frac{1}{n^2}\sum_{i,j\in I_n} \Var(A_1) = \Var(A_1).
     \end{split}
     \end{equation}
Equality in (\ref{Aiineq}) holds only if for all $i,j \in I_n$ $\Cov(A_i,A_j)$ is equal to $\Var(A_1)$, which
from Theorem \ref{thCov} occurs iff $A_i = A_j$.
\end{proof}
The above Theorem could also be proved using Schwarz inequality similarly as in Theorem 5 in \cite{Halmos_1946}.

After symmetrization of scheme given by (\ref{V1a3tot}) in the first argument we receive scheme given by
\begin{equation}\label{vTotMin}
\widehat{V}_{1,s2}^{tot}: = \frac{1}{2}(g[0][0] - g[1][0])^2.
\end{equation}
It uses the same number of function evaluations, so it has not higher inefficiency constant.
Analogously to what we did for total sensitivity index in the last section we can define generalized estimand
corresponding to main sensitivity index of a function with
some product distribution of arguments with respect
to a given argument and specify its domain. Since definitions of such generalized estimands are obvious, we omit them. 
Using Theorem \ref{thCond} one can receive the following well-known estimator for main sensitivity index with respect to the first argument
\begin{equation}
\widehat{V}_{1,a3}: = g[0][0](g[0][1] - g[1][1]).
\end{equation}
After symmetrization of its scheme in all coordinates we receive scheme given by
\begin{equation}\label{V1s4}
\widehat{V}_{1,s4}: =  \frac{1}{2}(g[0][0] - g[1][0])(g[0][1] - g[1][1]).
\end{equation}
Its estimator has not higher variance, but uses 4 rather than 3 function evaluations, so that their respective inefficiency constants
 fulfill
\begin{equation}\label{ineqV1}
d_{V_1,s4} \leq \frac{4}{3} d_{V_1,a3},
\end{equation}
which should be understood as relation valid for all appropriate $(f,\mu)$.
For $f(X_1, X_2) = X_1$ and $P(X_1 = 1) = P(X_1 = -1) = 1/2$ the variance of both estimators is equal to $1$, so
for inefficiency constant corresponding to such specified $f$ and any $\mu \sim (X_1, X_2)$ we have equality in (\ref{ineqV1}).

One can get estimator for total sensitivity index using the same function evaluations as (\ref{V1s4})
\begin{equation}\label{VTot1s4}
\widehat{V}^{tot}_{1, s4}: =  \frac{1}{4}\sum_{i=0}^1(g[0][i] - g[1][i])^2
\end{equation}
and also an estimator for variance of $f$
\begin{equation}
\widehat{D}_{s4}: = \frac{1}{4}\sum_{i=0}^1(g[0][i] - g[1][1-i])^2.
\end{equation}
One may wonder what is the relation between inefficiency constants of schemes for estimation of total
sensitivity index given by (\ref{vTotMin}) and (\ref{VTot1s4}). We receive it from the following Theorem.
\begin{theorem}
Let us consider scheme $(t_2, A_2)$ of some order $N$ created from $(t_1, A_1)$ by its symmetrization in the $i$-th
argument from 1 to 2 dimensions.
Then their inefficiency constants fulfill
\begin{equation}\label{d12}
d_{t_1, A_1} \leq d_{t_2, A_2} \leq 2 d_{t_1, A_1}.
\end{equation}
\end{theorem}
\begin{proof}
Let $\phi_i$ for $i \in I_2$ be estimators given by corresponding schemes. We have
$(A_1)_i = \{1\}$ and $(A_2)_i = \{1,2\}$. For some observable $\wt{X} \sim \mu^{n_{A_2}} $ we denote
$\wt{V}:= \wt{X}_{i}$ as in (\ref{wtxi}), and $\wt{U} :=(\wt{X}_j)_{j \in I_N, j \neq i}$. We have
\begin{equation}\label{phiSym2}
\phi_2(f)(\wt{X}) = \frac{1}{2}(\phi_1(f)(\wt{U}, \wt{V}_{1}) + \phi_1(f)(\wt{U}, \wt{V}_{2})).
\end{equation}
Taking variance of the rhs of (\ref{phiSym2}) we get
\begin{equation}
\frac{1}{2}(\Var_{f,\mu}(\phi_1) + \Cov(\phi_1(f)(\wt{U}, \wt{V}_{1}), \phi_1(f)(\wt{U}, \wt{V}_{2}))).
\end{equation}
From Theorem \ref{thCond} it follows that the last covariance is equal to
$\Var(\E(\phi_1(f)(\wt{U}, \wt{V}_1)|\wt{U}))$.
Thus from the fact that
\begin{equation}
0 \leq \Var(\E(\phi_1(f)(\wt{U}, \wt{V}_1)|\wt{U})) \leq \Var_{f,\mu}(\phi_1)
\end{equation}
variances of these estimators fulfill
\begin{equation}
\frac{1}{2}\Var_{f, \mu}(\phi_1) \leq \Var_{f, \mu}(\phi_2) \leq \Var_{f, \mu}(\phi_1).
\end{equation}
Since  $|A_2| = 2|A_1|$ from (\ref{dIneff}) we receive (\ref{d12}).
\end{proof}
Scheme (\ref{VTot1s4}) is received from (\ref{vTotMin}) by symmetrization in the 2-nd argument from 1 to 2 dimensions,
thus from the above Theorem we receive
\begin{equation}\label{compDi}
d_{V^{tot}_{1},s2} \leq  d_{V^{tot}_{1},s4}  \leq 2d_{V^{tot}_{1},s2}.
\end{equation}

\section{\label{secMany}Schemes for estimation of sensitivity indices with respect to many arguments}
%
We can be interested in estimating values of many estimands simultaneously, e. g. 
main variance-based sensitivity indices of output of a model with respect to all its 
parameters in order to decide which of them to measure experimentally, 
as discussed in Section \ref{varRedSec}. 
Let us consider MC procedure, in which $N$ different schemes equal to coordinates of $S = ((A_i, t_i))_{i=1}^N$ are used in a single MC step 
for estimation of $\lambda_1, \ldots, \lambda_N$. We call $S$ scheme for estimation of $\lambda_1, \ldots, \lambda_N$.
$(A_i, t_i)$ is called the sub scheme of $S$ for estimating $\lambda_i$. 
We define inefficiency constant of $S$ in estimating $\lambda_i$ for $i \in I_N$, as 
\begin{equation}\label{dInSch}
d_{\lambda_i, S}(f,\mu) = |\bigcup_{i=1}^N A_i| \Var_{f,\mu}(\phi(A_i,t_i)),
\end{equation}
where $\phi(A_i,t_i)$ is estimator corresponding to $(A_i,t_i)$. 
$|\bigcup_{i=1}^N A_i|$ is equal to the number of function evaluations required for computing estimates of all estimands, 
using this scheme. 
Constant (\ref{dInSch}) has similar interpretations as inefficiency constants of schemes for estimation of single estimands, as discussed in 
the previous Section. One could also use some measures of inefficiency of the scheme in estimating all estimands, e. g. certain weighted sum 
of the above defined inefficiency constants, but we further focus only on constants (\ref{dInSch}). 
Relations we derive here for (\ref{dInSch}) for different schemes can be used to derive similar relations for weighted sums. 

We now describe a scheme for estimation of all main and total sensitivity indices of functions $f(P)$ 
of some vector $P=(P_i)_{i=1}^{N_P}\sim \mu$ with independent coordinates (assuming $f(P) \in L^2(\PR)$), 
where the sensitivities are computed with respect to individual coordinates. 
This scheme was used in \cite{Zhang2009} for sensitivity analysis of a deterministic chemical kinetic model. 
Let $\wt{P}=(\wt{P}[i])_{i=0}^{1}$, where $\wt{P}[i] = (\wt{P}_j[i])_{j=1}^{N_P}$ for $i \in \{0,1\}$ are independent copies of $P$.
Let further $\wt{P}_k = (\wt{P}_k[i])_{k=0}^1$ and $\wt{P}_{(k)}[i]$ be equal to vector $\wt{P}[i]$ with $k$-th coordinate replaced by $\wt{P}_k[1-i]$.
For $i \in \{0,1\}$ we introduce helper functions 
\begin{equation}
s[i](f)(\wt{P}) = f(\wt{P}[i]),
\end{equation}
\begin{equation}
s_k[i](f)(\wt{P}) = f(\wt{P}_{(k)}[i]),
\end{equation}
which are just convenient notation for certain functions $g_{j}$ defined by (\ref{gjDef}). 
The scheme for estimation of main sensitivity index with respect to $k$-th argument in $O$ is given by
\begin{equation}\label{VkO}
\widehat{V}_{k,O} = \frac{1}{2}(s[0] - s_k[0])(s_k[1] - s[1]),
\end{equation}
while for the total sensitivity index by
\begin{equation}\label{VkTotO}
 \widehat{V}_{k,O}^{tot} = \frac{1}{4}\sum_{i=0}^{1}(s[i] - s_k[i])^2.
\end{equation}
One can also estimate a number of further indices using the same function evaluations, see \cite{Zhang2009} for schemes for estimation of 
sensitivity indices
with respect to pairs of parameters and for variance of $f(P)$. For $N_P > 2$ this scheme requires $2(N_P + 1)$ function evaluations, while for
$N_P = 2$ only $2N_P$, since we have $s_1[i] = s_2[1-i]$.
For $N_P=3$ one can find schemes with lower inefficiency constants in estimating all of these indices, 
 given that the original constants were nonzero. We discuss it in Appendix \ref{app3Params}.
For estimation of only total sensitivity indices one
 receives not higher inefficiency constant when using following estimator for each $k$-th parameter
\begin{equation}\label{VkTotOT}
 \widehat{V}_{k,OT}^{tot} = \frac{1}{2}(s[0] - s_k[0])^2.
\end{equation}
The relation between efficiency constants of schemes given by (\ref{VkTotO}) and (\ref{VkTotOT}) is, for all $ k \in I_{N_P}$ 
\begin{equation}
d_{V_k^{tot}, OT} \leq d_{V_k^{tot}, O} \leq 2d_{V_k^{tot}, OT},
\end{equation}
 since for fixed $k$ (\ref{VkTotO}) is received from (\ref{VkTotOT}) by symmetrization from 1 to 2 dimensions 
in the second argument if $f$ is treated as function of two arguments corresponding to values of random variables $P_k$ and $P_{\sim k}$. 

Let us now focus on functions of form $f(P,R)$, for $P$ as before and random variable $R$ independent of $P$. $f(P,R)$ can be for instance
observable of DM, corresponding to its certain construction and noise term $R$. Let us take $\wt{P}$ as before and
$\wt{R} = (\wt{R}[i])_{i=0}^1$, where $\wt{R}[i]$ are independent copies of $R$, independent of $\wt{P}$.
We now present different new schemes for estimation of both main and
total sensitivity indices of conditional expectation of $\E(f(P,R)|P)$ with respect to individual parameters,
pairs $(P_i,P_j)$ and a number of other indices. 
We define
\begin{equation}
s[i][j](f)(\wt{P},\wt{R}) = f(\wt{P}[i],\wt{R}[j])
\end{equation}
and
\begin{equation}
s_k[i][j](f)(\wt{P},\wt{R}) = f(\wt{P}_{(k)}[i],\wt{R}[j]).
\end{equation}
We use notation
\begin{equation}
\wt{s}[i][j] := s[i][j](f)(\wt{P},\wt{R})
\end{equation}
to denote the observable of this estimator corresponding to $f$ and variables $\wt{P}$ $\wt{R}$ and analogically for $\wt{s}_k[i][j]$. 
We also denote observable of a generalized estimator of $\lambda_i$ from scheme $S$ $\widehat{\lambda}_{i,S}(f)(\wt{P}, \wt{R})$  
simply as $\lambda_{i,S}$. 
For $f(P,R)$ we define indices like $V_{P_k}$, $\widetilde{V}^{tot}_{P_k}$, $D$ and $V_P$
 in the same way as in Section \ref{sensAveSec} for observables of DM. 
The fact that schemes for estimation of individual indices given by formulas below
are unbiased is an easy consequence of Theorem \ref{thCond} and expressions for respective sensitivity indices derived in Section \ref{sensAveSec}.
We first define a C language-like notation
\begin{equation}\label{CNot}
(a == b)?c:d = \begin{cases}
      c & \text{ if $a=b$,} \\
      d & \text{otherwise,} \\
      \end{cases}
\end{equation}
and helper functions
\begin{equation}\label{AM}
A_M(s,l,r) = \frac{1}{4} \sum_{i=0}^{1} \sum_{j=0}^{1} s[i][j] s[(l==1)?i:(1-i)][(r==1)?j:(1-j)],
\end{equation}
\begin{equation}
B_M(s,s_k,l,r) = \frac{1}{4} \sum_{i=0}^{1} \sum_{j=0}^{1} s[i][j] s_{k}[(l==1)?i:(1-i)][(r==1)?j:(1-j)].
\end{equation}
For $k \in I_{N_P}$ we define
\begin{equation}\label{estVkE}
\begin{split}
\widehat{V}_{k,E} &:= B_M(s,s_k,0,0) - \frac{A_M(s,0,0) +  A_M(s_k,0,0)}{2} \\
& = \frac{1}{4}\sum_{i=0}^{1}(s[i][0] - s_k[i][0])(s_k[1-i][1] - s[1-i][1]),
\end{split}
\end{equation}
\begin{equation}\label{estVkTotE}
\begin{split}
\widehat{\widetilde{V}}^{tot}_{k,E} &:= \frac{A_M(s,1,0)+  A_M(s_k,1,0)}{2} -  B_M(s,s_k,1,0)\\
 &= \frac{1}{4} \sum_{i=0}^{1}(s[i][0] - s_{k}[i][0])(s[i][1] - s_{k}[i][1]).
\end{split}
\end{equation}
\begin{eqnarray}\label{VEst}
\begin{split}
\widehat{D}_{E} := \frac{1}{(N_P + 1)}(A_M(s,1,1) - A_M(s,0,0) +  \sum_{k=1}^{N_P}A_M(s_k,1,1) - A_M(s_k,0,0)),
\end{split}
\end{eqnarray}
\begin{eqnarray}\label{VP}
\widehat{V}_{P,E} := \frac{1}{(N_P + 1)}(A_M(s,1,0) - A_M(s,0,0) + \sum_{k=1}^{N_P}A_M(s_k,1,0) - A_M(s_k,0,0)),
\end{eqnarray}
\begin{eqnarray}\label{VRtot}
\begin{split}
\widehat{V}_{R,E}^{tot} := \widehat{D}_{E} - \widehat{V}_{P,E}.
\end{split}
\end{eqnarray}
Using the same function evaluations we can also construct schemes for estimation of many further indices, among others 
for $V_{(P_{i},P_{j})}$ and $\widetilde{V}_{(P_{i},P_{j})}^{tot}$, which we describe in Appendix \ref{appFurthEstims}.   
 The scheme for estimation of sensitivity indices of conditional expectation consisting of individual schemes
given by formulas above is called scheme E. 
It uses $4(N_P + 1)$ function evaluations for $N_P > 2$ and  $4N_P$ for $N_P = 2$, since in the last case we have
\begin{equation}\label{equEvals}
s_2[j][i]= s_1[1-j][i].
\end{equation}
For $N_P = 3$ there exists more efficient scheme as discussed in Appendix \ref{app3Params}. 

Scheme EM, which can have lower inefficiency constants in estimating main sensitivity indices of conditional expectations consists of sub schemes 
given by, for $k \in I_{N_P}$ 
\begin{equation}\label{estVkEM}
 \widehat{V}_{k,EM} = \frac{1}{2}(s[0][0] - s_k[0][0])(s_k[1][1] - s[1][1]).
\end{equation}
Scheme EM uses two times fewer function evaluations than E, for $N_P > 2$.
In Appendix \ref{appFurthEstims} we prove the following Theorem.
\begin{theorem}\label{thdEMEComp}
Inefficiency constants of scheme EM and E  for estimation of $V_i$ fulfill, for $N_P > 2$
\begin{equation}\label{EMComp}
d_{V_i,EM} \leq d_{V_i,E} \leq 2d_{V_i,EM}.
\end{equation}
\end{theorem}
Due to proportionality of number of function evaluations used by schemes for all indices and schemes for individual indices,
the same inequalities hold also for the latter. 
We can extend this scheme to scheme EMe which uses additional functions in $(s[i][1-i])_{i=0}^1$ 
in sub schemes for total sensitivity indices of the mean, for $k \in I_{N_P}$
\begin{equation}
\widehat{\widetilde{V}}_{k, EMe}^{tot} := \frac{1}{2}\sum_{i=0}^1(s[i][i]s[i][1-i] - s_k[i][i]s[i][1-i]).
\end{equation}

Scheme ET, which can be potentially more efficient for estimation of total sensitivity indices, 
contains schemes 
\begin{equation}\label{estVkTotET}
 \widehat{\widetilde{V}}^{tot}_{k,ET} := \frac{1}{2}(s[0][0] - s_k[0][0])(s[0][1] - s_k[0][1]),
\end{equation}
for $k \in I_{N_P}$, which use together two times fewer function evaluations than scheme E, for $N_P>2$.
We have the following relations
\begin{equation}\label{VitotComp}
d_{\wt{V}_i^{tot},ET} \leq d_{\wt{V}_i^{tot},E} \leq 2d_{\wt{V}_i^{tot},ET}
\end{equation}
and analogically for inefficiency constants of their sub schemes for estimating these indices.
This is a consequence of the fact that if arguments of $f$ corresponding to coordinates of $\wt{P}_{\sim k}[0]$ are treated as a single 
argument, then scheme defining $\wt{V}_{i,E}^{tot}$ is symmetrization of scheme defining 
 $\widetilde{\wt{V}}_{i,ET}^{tot}$ from $1$ to $2$ dimensions in this argument.
We can extend scheme ET to ETe by adding to it sub schemes for estimation of main sensitivity index of conditional mean, given by formula 
\begin{equation}\label{VKETe}
\widehat{\wt{V}}_{k, ETe} := \frac{1}{2}\sum_{i=0}^1(s_k[0][i]s[1][1-i]) - A_M(s,0,0),
\end{equation}
which additionally needs functions $(s[i][i])_{i=0}^1$. 
The number of function evaluations used by schemes EMe and ETe is $\frac{N_P +2}{N_P +1}$ times this number for schemes EM and ET. 
Since scheme defining $\widehat{\wt{V}}_{k, E}$ is symmetrization of $\widehat{\wt{V}}_{k, ETe}$ in the argument corresponding to 
$P_k$ in $f(P,R)$, we receive the following relationship for $N_P > 2$
\begin{equation}
2\frac{N_P +1}{N_P +2} d_{V_k, ETe} \geq d_{V_k, E}
\end{equation}
and analogically for schemes defining $\widehat{\wt{V}}^{tot}_{k, E}$ which is symmetrization of scheme defining $\widehat{\wt{V}}^{tot}_{k, EMe}$ 
with respect to the same group. In numerical examples we will see that $d_{V_k, ETe}$ can be much higher than 
 $d_{V_k, E}$ and analogously for  $d_{\wt{V}^{tot}_k, EM}$ and $d_{\wt{V}^{tot}_k, E}$. 
Let us also notice, that for functions of additive form 
\begin{equation}\label{addForm}
f(P,R) = f_1(P) + f_2(R)
\end{equation}
observables of estimators  (\ref{estVkE}) and (\ref{estVkEM}) corresponding to the same $\wt{P}$ are equal for every $k \in I_{N_P}$. 
In particular they have the same variances and we have equality in the right inequality of relation 
(\ref{EMComp}). If further for $f_1(P)$ from (\ref{addForm}) it holds 
\begin{equation}
f_1(P) = \sum_{i=1}^{N_P} f_{1,i}(P_i)
\end{equation}
then also appropriate observables of 
estimators (\ref{estVkTotE}) and (\ref{estVkTotET}) are identical and we have equality in the right inequality of relation (\ref{VitotComp}). 

All the schemes for estimation of sensitivity indices of conditional expectations introduced in this Section
 can be also used for conditional histograms, 
except that instead of using real-valued observables one should use their vector-valued 
single-sample histograms and instead of function multiplication use scalar 
product of vectors. This is a consequence expression (\ref{thCondVect}) after the proof of Theorem \ref{thCond}.

\section{Variances of estimators for different constructions of DM}\label{secVarDiff}
If two estimators are unbiased the relation between their variances is the same as between 
the expectations of their squares. Let $h(p,R)$ be certain construction of DMCP  (\ref{DMCPFun}) and $f(p,R) = g(h(p,R))$ 
its observable. 
In each step of MC simulations performed in \cite{Rathinam_2010} the values of independent copies of an observable $f(p,R)$, representing 
the number of particles of certain species at a given moment of time, were generated for the nominal parameter value $p$ and 
the values of copies of observable $f(p+he_i,R)$ for some small perturbation $h$ of the $i$-th coordinate of $p$, for the purpose 
of estimating finite differences of means 
\begin{equation}
\frac{1}{|h|} \E(f(p + he_i, R) - f(p, R)).
\end{equation}
The computed estimates of the following expectation
\begin{eqnarray}
 \err(p, p + e_ih) = \E((f(p, R) - f(p + e_ih, R))^2),
\end{eqnarray}
which influences the variance of estimators they used, were much lower when performing simulations with RTC construction, rather than GD construction. 
We call this effect tighter coupling between the value of the considered observable to the noise term for the nominal and perturbed values 
of parameters for RTC than for GD algorithm. Reader is referred to the original work \cite{Rathinam_2010} for a number of intuitions supporting this 
effect. 
Chemical reaction networks for which this effect was observed contained reactions influencing the investigated particle numbers 
in different ways, for instance in some reactions the number of particles increased, while in others it decreased. 
Let us assume that for such reaction networks, 
 for all $p_1$ and $p_2$ in the image of $P$, $\err(p_1, p_2)$ is greater for the same observables constructed using GD than RTC method. 
Since for an observable of estimator (\ref{estVkEM}) we have, using notations from previous Section 
  \begin{equation}
  \begin{split}
  4\E((V_{k,EM})^2) &= \E((\wt{s}_k[0][0] -\wt{s}[0][0])^2(\wt{s}[1][1] - \wt{s}_k[1][1])^2) \\
 & = \E((\E^2(f(p_{k,0}, \wt{R}[0]) - f(p_{0}, \wt{R}[0])) \\
 & \cdot \E^2(f(p_{1},\wt{R}[1]) -f(p_{k,1},\wt{R}[1])))_{p_j=\wt{P}[j], p_{k,j}= \wt{P}_{(k)}[j], j \in\{0, 1\}})\\
  & = \E(\err(\wt{P}_{(k)}[0], \wt{P}[0])\err(\wt{P}[1], \wt{P}_{(k)}[1])),\\
  \end{split}
 \end{equation}
so with the assumption made this should be greater for GD than RTC construction. 
Although we could not confirm whether this assumption is true, the decrease of estimated variance of estimator (\ref{estVkEM}) was indeed confirmed
in all our numerical experiments involving chemical reaction networks containing reactions influencing particle numbers in different ways.
Denoting $A[i] = \frac{1}{4}(\wt{s}[i][0] - \wt{s}_k[i][0])(\wt{s}_k[1-i][1] - \wt{s}[1-i][1]))$ we have that 
an observable of estimator (\ref{estVkE}) fulfills
\begin{equation}
\E((V_{k,E})^2) = \E((\sum_{i=0}^1 A[i])^2)
= 2(E((V_{k,EM})^2) + \Cov(A[0],A[1])).
\end{equation}
The estimates of $\Cov(A[0],A[1])$ from our numerical experiments sometimes decreased and sometimes increased when going
from GD to RTC construction, but we nevertheless always observed the decrease of the estimated value of $\E((V_{k,E})^2)$.
Note that observable of estimator (\ref{estVkTotET}) fulfills 
\begin{equation}
 4\E{((\widetilde{V}}^{tot}_{k,ET})^2) = \E(\err(\wt{P}_{(k)}[0], \wt{P}[0])^2),
\end{equation}
so we could suspect it should also have lower variance for RTC than GD method and this was 
indeed confirmed in our numerical experiments. Intuitions and numerical results 
for the estimator of total sensitivity index of mean from scheme E (\ref{estVkTotE}) 
were analogical as in case of the main index. 

Note that although change of order of the indices of reactions in a 
chemical reaction network does not influence the 
variance of the estimators using RTC construction, as only reorders the Poisson processes in the construction, 
it might have impact on the variance of estimators when using GD algorithm. 
We had an intuition that grouping reactions having similar influence on the output together in the sequence of reactions used by GD construction 
 should lead to tighter coupling between the observables and the noise term for different values of parameters, and thus to lower variances of
estimators we discussed in this Section, than when reactions with opposing effects
appear in the sequence one after another. This is because we suspected that
 reactions lying close to one another in the sequence may often be fired
in the same step of constructions using two different values of parameters and the same noise term.
We will see this effect confirmed in Section \ref{secMBMD}, in a numerical experiment specially designed for testing it.

\section{Implementation}\label{secImpl}
All our numerical experiments were performed using a program written in C++ language, run on a
 personal computer with 1GB RAM, 2-core 2.10 Hz processor and with Linux operating system.
 For random number generation we used Gnu Scientific Library (GSL) \cite{Galassi_2003} implementation of
Mersenne twister random number generator (RNG) \cite{MATSUMOTO_NISHIMURA_1998}.
Using notations from Section \ref{secMany}, at the beginning of each Monte Carlo step we generated
value of an independent copy of a variable $\wt{P}$. 
Observables of functions needed by a a given scheme 
 were generated by running given simulation algorithm starting with appropriate
parameters and reusing the same 
generated values of artificial noise variable $\wt{R}[j]$ to compute values of observables  $\wt{s}[i][j]$ and $\wt{s}_{(k)}[i][j]$ for the
 same $j \in \{0,1\}$.
We describe different strategies for reusing values of these variables later on in this Section. 

We used simple implementations of GD method and RTC algorithm, that is we used arrays to store
reaction states and propensities as well as linear search for minimum to obtain $S_{i+1}$ in RTC construction,
or reaction to fire in GD method.
 Often simplest implementations turn out to be the fastest when simulating small reaction networks, whereas smaller computation
 time can be achieved for more complex networks when using improvements like dependency graphs, sparse arrays, priority
 queues or dynamical reordering of
reactions in GD method \cite{McCollum_2006}. See  \cite{Mauch_2011} for a recent review and comparison of
computation times of simulations using a variety of different data structures. 
 Most of these enhancements
can be incorporated into our algorithms without changing the variance of resulting estimators, but some, like dynamically changing the order of
reactions in GD algorithm may have impact on the variance.

We numerically investigated two different approaches to reusing values of each independent copy of the noise variable needed in a single Monte Carlo step.  
Similar methods were suggested in  \cite{Rathinam_2010} for performing local sensitivity analysis, where, however, only
the first method was tested numerically. In the first method, in addition to the main RNG used for generating parameters,
 one uses separate RNGs for simulating the noise variables, one
RNG in the GD method and one RNG for every independent Poisson process in the RTC method. In order to retrieve
the same values of artificial noise variable $\wt{R}[j]$ one reuses the same initial seeds of RNGs for noise variables, which are chosen randomly
at the beginning of the Monte Carlo step from the possible RNG seed range, using the main RNG. The drawback of this method is that one needs to 
 generate the same random numbers and to reinitialize RNGs for noise variables several times in each Monte Carlo step.
 In the second approach one stores the values of the same independent noise term in a separate set of lists.
In GD method one uses single list for every noise variable, while in RTC method different one for every Poisson process.
New values are added to the lists when needed and new memory is allocated to make the list longer
only when more random numbers are produced for a given list than in previous Monte Carlo steps.
The disadvantage of this method is that one needs additional memory
for the lists and uses up time for reading from and writing onto them.
Fortunately, we did not experience any memory exhaustion problem in our numerical experiments.

We compared the average execution time of MC procedures using scheme E from Section \ref{secMany}
 on three chemical reaction networks,
which we define in Chapter \ref{chapNumExp} and for the two mentioned approaches to reusing the values of noise variables.
 The results are presented in Table \ref{2Appr}. We can see that the first approach was from $5\%$ to $89\%$ slower.
When the reinitialization of RNG
was commented out
 we observed that the first approach was on average only from $4\%$ to $7\%$ slower in all examples (data not shown), thus  
 high relative increase of cost for instance in the MBMD model
 can be explained by the significant contribution of the RNG reinitialization
 to the overall cost of a single simulation due to the single simulation being rather short.
The data presented in the further Sections was produced using only the second approach.
\begin{table}[h]
\resizebox{16cm}{!} {
\begin{tabular}{|l|c|c|c|c|c|c|}
\hline
DM   & $\text{RTC}_l$ & $\text{RTC}_{nl}$ & $\%$ inc.&$\text{GD}_l$& $\text{GD}_{nl}$ &$\%$ inc.\\
\hline
SB &$97.13 \pm 0.49$ &$111.8\pm1.7$ & $15\%$ & $103.65\pm0.60$& $122.07\pm0.71$& $18\%$\\ 
GTS & $130.73\pm 0.48    $&$137.17 \pm 0.26$&$    4.9\%$&$    132.07\pm 0.10$&    $138.94 \pm 0.22$&$    5.2\%$\\
MBMD &$110.32\pm 0.04$& $208.02\pm0.86$& $89\%$ & $107.60\pm0.09 $&$129.66\pm0.56 $ & $21\%$\\ 
\hline
\end{tabular}
}
\caption{\label{2Appr} The Table presents mean execution times in seconds of MC procedures using scheme E 
for observables of simple birth (SB), genetic toggle-switch (GTS) and many births - many deaths (MBMD) models 
 defined in Chapter \ref{chapNumExp}.
RTC and GD algorithms were used both with and without lists for reusing artificial noise variables (denoted by subscripts $l$ and $nl$ respectively).
$50000$ steps were performed for SB and MBMD models and $5000$ for GTS. 
The means were computed from 3 runs with random initial seeds and are given with estimates of standard errors. 
The ``$\%$ inc.'' column contains the relative increases of estimated mean duration times of MC procedure without lists
over the one with lists.}
\end{table}

From Table \ref{2Appr} we can also see that the execution times of simulations using GD and RTC methods and the
approach with lists were approximately the same for all models. 

\section{\label{secQMC}Quasi-Monte Carlo and variance reduction methods}
 One can often speed up the computation of integrals by using quasi-Monte Carlo (QMC) or different variance-reduction techniques rather than
ordinary MC method.
In QMC method one generates vectors from $n$-dimensional cube $[0,1]^n$ for some fixed $n$ in each step of the method
using a quasi-random number generator (QRNG)
and uses them in the same way as values of observables in form of random vectors $U=(U_i)_{i=1}^{N_P}$ with
independent coordinates with distribution $\U(0,1)$ in ordinary MC.
For DM with independent parameters $P \sim \mu_1 \times \ldots \times \mu_{N_P} $ one can often find functions $g_i$ for $i \in I_{N_P}$, such that
$(g_i(U_i))_{i=1}^{N_P} \sim P$. For instance for the case of $P_i \sim \U(a_i,b_i)$ one can use
\begin{equation}
g_i(x) = x(b_i-a_i)+a_i,
\end{equation}
while for $P_i \sim U_d(a_i,b_i)$
\begin{equation}
g_i(x) = \lfloor x (b_i-a_i+1)\rfloor + a_i,
\end{equation}
where for $x \in \R$ its floor $\lfloor x \rfloor$ denotes the biggest integer smaller than or equal to $x$.
We use the QMC for the parameters and MC for
the noise variable approach, which relies on using vectors generated with the help of QRNG
and above functions to sample parameters and ordinary RNG to sample the artificial noise variable.
We call it hybrid QMC-MC approach. One could use QRNG
to sample certain number of components of the artificial noise variable as well, but not this whole variable, as
we do not know in advance how many of its components are needed in a simulation. A problem with using QMC is that although it usually leads
to smaller error than MC
there are currently no well-established methods for obtaining error estimates of the result from simulation data \cite{Owen05onthe}.
Such error estimates can be obtained by variance-reduction techniques like randomized-quasi Monte Carlo \cite{Owen05onthe},
 which we however do not test in this work.
\section{The method of Degasperi et al.}\label{secDega}
We now present generalization of method of Degasperi et al. \cite{Degasperi2008}, which we already mentioned in the Introduction and Section
\ref{secOrthog}.
Let us assume that DM has $N_P$ independent parameters $P=(P_i)_{i=1}^{N_P}$, which can be mapped from variables
with distributions $\U(0,1)$ as discussed in the previous Section. Replacing original parameters with these mapped variables
when necessary, we assume that $P_i \sim \U(0,1)$ for $i \in I_{N_P}$.
For $m$ natural positive, called grid level,
 we define discretized parameter vector $D$, also called parameter grid, as a
function $D = (D_i)_{i=1}^{N_P}:I_m^{N_P} \longmapsto \R^{N_P}$, whose coordinates,
 called discretized parameters, satisfy for any $j = (j_i)_{i=1}^{N_P} \in I_m^{N_P}$
\begin{equation}
D_{i}[j] = \frac{j_i}{m + 1}.
\end{equation}
Let us assume we estimate sensitivity indices associated with conditional expectation $\tilde{f}(P)$
for some observable $f(P,R) \in L^2(\PR)$. For some $N_s$ natural positive, called number of simulations
in each point of the grid, and independent copies of noise variable $R[j][k] \sim R$ for every
$j \in I_m^{N_P}$ and $k \in I_{N_s}$ one computes in a simulation the value of $f(D[j], R[j][k])$.
Then one computes discretized estimate of $\tilde{f}(P)$ for every $j \in I_m^{N_P}$
\begin{equation}\label{fP}
{\tilde{f}}(D[j]) = \phi_{ave}\left(\left(f(D[j], R[j][k](\omega)\right)_{k \in I_{N_s}}\right).
\end{equation}
Now one treats $D$ as a random vector on the discrete probability space
$I_m^{N_P}$ with equal probabilities of its elements and performs exact variance-based sensitivity analysis on function
$\tilde{f}(D)$. Firstly, one computes conditional expectations of $\tilde{f}(D)$ given certain sub vectors of $D$, for $J \subset I$
 and $v_J \in I_m^{J}$
\begin{equation}
\E(\tilde{f}(D)|D_J)(v_J) = \frac{1}{m^{|\sim J|}} \sum_{j_{\sim J} \in I_m^{\sim J}} \tilde{f}(D[v_J, j_{\sim J}]).
\end{equation}
Secondly, one computes variances of these conditional expectations needed
to calculate the desired variance-based sensitivity indices of $\tilde{f}(D)$ with respect sub vectors of $D$.
They are used to approximate the sensitivity indices of $\tilde{f}(P)$ given the corresponding sub vectors of $P$.
When approximating sensitivity indices associated with  conditional histograms, the procedure is
the same, except that one should use
unbiased estimator or average histogram, like mean of histogram functions instead of $\phi_{ave}$ in (\ref{fP}).
Degasperi et. al. used this method for computing variance based-sensitivity indices of conditional histograms
using variance defined with Manhattan distance as discussed in Section \ref{secOrthog}.
In the next Section we apply this method to conditional expectations with standard variance and call it shortly \mbox{grid-based} method.
Unfortunately, the method presented here does not provide error estimates for the computed approximations of sensitivity indices.
\chapter{Numerical experiments}\label{chapNumExp}
\section{Simple birth model}\label{secMB}
We first apply presented methods to a simple birth (SB) model, for which analytical expressions for most of the \mbox{variance-based}
sensitivity indices can be obtained. This allows for estimation of mean square errors of approximation of variance-based sensitivity indices
using grid-based method described in Section \ref{secDega} and hybrid QMC-MC approach discussed in Section \ref{secQMC}.
The reaction network contains one species $X$ and one birth reaction can occur
\begin{equation}
R_1:\ \emptyset \longrightarrow  X.
\end{equation}
It is described by a kinetic formula $a_1(K)(x) = K_1 + K_2 + K_3$, where $K = (K_1,K_2,K_3)$ is a random vector
with independent coordinates with uniform respective marginal distributions $U(0.3,0.9)$, $U(0.85, 1.15)$ and
$U(0.07, 0.13)$. Variable $C$ describing initial number
of particles of species $X$ has uniform discrete distribution $U_d(30, 90)$ and is independent of $K$.
Note that for this reaction network as well as for the ones in further numerical examples 
assumptions of Theorem \ref{thDMCP} are satisfied for every values of parameters when we take $m_i$ equal to 1 for every
$i$-th species.
The observable taken for sensitivity analysis is the number of species $X$ at time $t=100$.
In Appendix \ref{appd} we derive analytical expressions for some of the \mbox{variance-based} sensitivity indices of this observable and its conditional
expectation given the parameters.
The values obtained from analytic expressions are presented in Table \ref{exactVal}.
For a reaction network with one reaction there is no difference in variance of estimators using GD and RTC algorithms and
we use only the latter one.

We performed a 50000 step MC procedure using scheme E. The results, presented in Table \ref{tabMBRTC} are in good agreement with analytically
computed values from Table \ref{exactVal}.
We also performed computations with grid-based method with grid level $10$ and $100$
 simulations in every point of the grid, which also requires one million simulations in total. From the results in table
\ref{tabGrid} we can see, that although the ordering of values of sensitivity indices computed with this method is correct,
these values are much farther from exact ones than for MC procedure. 

We performed an experiment comparing mean square errors of grid-based method and MC procedure using scheme E as before,
hybrid QMC-MC method using scheme E in $50000$ steps and MC procedures using schemes EMe and ETe in $100000$ steps.
For quasi-Monte Carlo sampling in the parameter space in the hybrid QMC-MC
approach we show only data obtained using the Niederreiter quasi-random sequences \cite{Bratley_1992}, but using
Sobol or Helton sequences, all available from GSL \cite{Galassi_2003}, lead to approximately the same results.
All methods except for MC procedures using schemes EMe and ETe involved $1$ million process simulations and had approximately the same duration.
MC using shemes EMe and ETe involved $1.2$ times more simulations, but when we consider only their sub schemes EM and ET,
they used 1 million simulations as well. Thus from discussion in Section \ref{secUnbiased}
 the ratio of variances of final MC estimators for all sensitivity indices using
schemes E and sub schemes EM and ET are the same as of their respective inefficiency constants (\ref{dIneff}).
We run each method a number of times using
the same random number generators without reinitializing, but for hybrid method reinitializing each time the QRNG. In every $i$-th step we
computed an estimate $\err_i$ of mean square error of a given method. For MC methods the estimate of $\err_i$ was estimate
of variance of final MC estimator
 (\ref{varEst}). For grid-based and hybrid QMC-MC methods, in which the estimate of sensitivity index with analytically computed value $\lambda$ in
the step was $D_i$, we took
\begin{equation}
\err_i = (D_i - \lambda)^2.
\end{equation}
The estimates of mean square error of each method and standard deviation of the result were obtained using mean (\ref{phiAve})
 and variance of mean (\ref{phiAveVar}) estimators applied to sequence of errors from all steps as in MC method.
Each method was run $5$ times, except for hybrid method, which was run $50$ times due to relatively high estimated values of
standard deviation of its samples of mean square error. The results are given in Table \ref{tabCompareMB}.
We can see that the estimated mean square error of grid-based method for estimation of sensitivity indices
is about three orders of magnitude higher than for pure MC methods using  E for all indices, EM for main and ET for
total sensitivity indices.
The estimate of mean square error of the hybrid method in approximating $V_C$ is even about $50$ million times
lower than for the grid-based method, but for index $V_{K_3}$ hybrid method does not seem to have any advantage over ordinary
MC using scheme E. 
It can be seen that computed variances of final MC estimators using schemes EMe
for total and ETe for main sensitivity indices are much higher than variances for estimators from other schemes and for the same indices, even though
the latter used fewer process simulations. Notice also, that estimates of variance of final MC estimators given by
scheme EM are approximately two times lower than for scheme E for all main indices
and similarly for schemes E and ET for total indices. This coincides with equalities in the right inequalities of  
relations (\ref{EMComp}) and (\ref{VitotComp}).
The estimated mean value of the model output and different
measures of its dispersion we discussed in Section \ref{secDisp} are given in Table \ref{tabDisp}, along with these values
for models from the following Sections.

%

\begin{table}[h]
\begin{tabular}{|l|c|c|c|c|}
\hline
$i $ & $\widetilde{V}_i$ & $\widetilde{V}_i^{tot} $ &$\widetilde{S}_i$ &$\widetilde{S}_i^{tot}$\\
\hline
$C $   &$ 310 $& $310$  &$0,451$&$0,451$ \\
$K_1 $ &$ 300 $& $300$ & $0,436  $ & $0,436 $\\
$K_2 $ &$ 75 $& $75$ & $0,109  $ & $0,109  $\\
$K_3 $ &$ 3 $& $3$ & $0,0044 $ & $0,0044 $\\
\hline
$i$& $V_i$ & $V_i^{tot} $ & $S_i$ & $S_i^{tot}$ \\
\cline{1-5}
$P $ & $ 688 $& & $0.80$ &  \\
$R $ & &$170$  & &$0.20$\\
$P,R $ &$ 858 $& $ 858 $ & $1$ & $0$ \\
\hline
\end{tabular}
\caption{\label{exactVal}Values of sensitivity indices
in the SB model obtained from analytic formulas derived in Section \ref{appd}}
\end{table}

\begin{table}[h]
\begin{tabular}{|l|c|c|c|c|}
\hline
$i $ & $\widetilde{V}_i$ & $\widetilde{V}_i^{tot} $ &$\widetilde{S}_i$ &$\widetilde{S}_i^{tot}$\\
\hline
$C $& $312.6 \pm 1.6$ & $312.6 \pm 1.6$ &$ 0.45 $&$ 0.45$\\ 
$k1 $& $301.4 \pm 1.6$ & $301.4 \pm 1.6$ &$ 0.44 $&$ 0.44$\\ 
$k2 $& $74.04 \pm 0.41$ & $74.03 \pm 0.41$ &$ 0.11 $&$ 0.11$\\ 
$k3 $& $3.003 \pm 0.020$ & $2.992 \pm 0.020$ &$ 0.0043 $&$ 0.0043$\\ 
\hline
$i $ & $V_i$ & $V_i^{tot} $ & $S_i$ & $S_i^{tot}$ \\
\hline
$P $& $690.4 \pm 2.5$ & $701.6 \pm 2.5$ &$ 0.8 $&$ 0.82$ \\ 
$R $& $159.2 \pm 1.0$ & $170.4 \pm 1.0$ &$ 0.18 $&$ 0.2$ \\
$P,R$ & $860.8 \pm 2.7$ & $860.8 \pm 2.7$ &$ 1 $&$ 1$ \\
\hline
\end{tabular}
\caption{\label{tabMBRTC} Estimates of sensitivity indices for the SB model computed in a 50000 step MC procedure using RTC algorithm
and scheme E.}
\end{table}

 \begin{table}[h]
  \begin{tabular}{|l|c|c|}
  \hline
  $i$ & $\widetilde{V}_i$ & $\widetilde{V}^{tot}_i$  \\
  \hline
 $X $&$ 250.8$&$ 252.3$\\ 
 $K_1 $&$ 245.9$&$ 247.5$\\ 
 $K_2 $&$ 61.29$&$ 62.87$\\ 
 $K_3 $&$ 2.393$&$ 3.97$\\ 
 $P $&$ 562.1$&$  562.1$\\ 
 \hline
  \end{tabular}
  \caption{\label{tabGrid} Estimates of sensitivity indices for SB model computed
  using grid-based method with $100$ simulations in every point of a
  level $10$ grid.}
 \end{table}

\begin{table}[h]
\resizebox{16cm}{!} {
\begin{tabular}{|l|c|c|c|c|c|c|}
\hline
{\multirow{2}{*}{i}}
 & Grid & E & E-QMC & EMe  & ETe\\
\cline{2-6}
& \multicolumn{5}{c|}{$\err \widetilde{V}_i$} \\
\cline{2-6}
\hline
$C$ &$3464 \pm 23$ &$2.6922 \pm 0.0061$ &$7.1 \pm 1.6\cdot 10^{-5}$& $1.3437 \pm 0.0019$ &$333.29 \pm 0.38$ \\
$K_1$ &$2986 \pm 19$ &$2.696 \pm 0.012$ &$0.224 \pm 0.046$& $1.3687 \pm 0.0053$ &$327.15 \pm 0.18$ \\
$K_2$ &$184.9 \pm 2.1$ &$0.1754 \pm 0.0016$ &$0.168 \pm 0.029$& $0.09212 \pm 0.00037$ &$83.420 \pm 0.065$ \\
$K_3$ &$0.293 \pm 0.022$ &$3.839 \pm 0.037\cdot 10^{-4}$ &$4.49 \pm 0.78\cdot 10^{-4}$&$2.490 \pm 0.010\cdot 10^{-4}$ &$3.7732 \pm 0.0064$\\
\hline
i & \multicolumn{5}{c|}{$\err \widetilde{V}^{tot}_i$} \\
\hline
$C$ &$3287 \pm 21$ &$2.6922 \pm 0.0061$ &$7.1 \pm 1.6\cdot 10^{-5}$& $3.0566 \pm 0.0086$ &$1.3474 \pm 0.0039$ \\
$K_1$ &$2821 \pm 19$ &$2.696 \pm 0.012$ &$0.224 \pm 0.045$ & $7.665 \pm 0.019$ &$1.3684 \pm 0.0030$ \\
$K_2$ &$145.6 \pm 1.9$ &$0.1754 \pm 0.0016$ &$0.162 \pm 0.028$& $3.2002 \pm 0.0087$ &$0.09266 \pm 0.00016$ \\
$K_3$ &$0.984 \pm 0.033$ &$3.842 \pm 0.033\cdot 10^{-4}$ &$4.69 \pm 0.76\cdot 10^{-4}$&$0.56047 \pm 0.00075$ &
$2.498 \pm 0.012\cdot 10^{-4}$ \\
\hline
i & \multicolumn{5}{c|}{$V_i$ } \\
\hline
$P$ &$15362 \pm 42$ &$6.033 \pm 0.023$ &$0.519 \pm 0.081$&$8.514 \pm 0.038$ &$8.585 \pm 0.016$ \\
\hline
\end{tabular}
}
\caption{\label{tabCompareMB} Estimates of mean square errors of final MC estimators of sensitivity indices of the SB
model computed using grid-based method (Grid), schemes E using MC procedure (E)
and hybrid approach (E-QMC) and schemes EMe and ETe in MC procedure. The errors are given along with their
estimated standard deviations (see main text in Section \ref{secMB} for details).}
\end{table}

\begin{table}[h]
\begin{tabular}{|l|c|c|c|c|c|}
\hline
DM& Mean & AveVar & G$\sigma$ & GFF & GCV \\
\hline
SB & $230.069 \pm 0.092$ & $170.4 \pm 1.0$&$ 13$ & $0.74 $&$ 0.057$ \\
MBMD &$10.0361 \pm 0.0099$ & $7.012 \pm 0.038$&$ 2.6$ & $0.7 $&$ 0.26$ \\
GTS &$30.299 \pm 0.060$ & $368.4 \pm 1.6$&$ 19$ & $12 $&$ 0.63$\\
\hline
\end{tabular}
\caption{\label{tabDisp} Estimates of means and different measures of dispersion like mean variance (AveVar), generalized
standard deviation (G$\sigma$), Fano Factor (GFF) (\ref{FFGen}) and coefficient of variation (GCV) (\ref{CVGen}) computed for SB,
MBMD and GTS models in 50000 step MC procedures using scheme E and RTC algorithm.}
\end{table}


\section{Genetic toggle switch model}\label{secGTS}
We now deal with a more biologically interesting model of a genetic toggle switch (GTS). It is
a simplified stochastic version of a model of a synthetic genetic toggle switch from  \cite{Gardner2000},
which was introduced and used for local sensitivity analysis in
 \cite{Rathinam_2010}.
The toggle switch consists of two promoters and their
respective repressors $U$ and $V$.
Each promoter is inhibited by a repressor transcribed by
the opposing promoter. The following reactions can occur
\[R_1:\ \emptyset \longmapsto  U, \quad R_2:\ U \longmapsto \emptyset,\]
\[R_3:\ \emptyset \longmapsto V,\quad R_4:\ V \longmapsto \emptyset.\]
Denoting $x = (x_1,x_2)$ the vector of numbers of species $U$ and $V$ respectively, the propensities
of the above reactions can be written as
\[a_{1}(x) = \frac{\alpha_1}{1 + x_2^{\beta}}, \quad a_{2}(x) = x_1, \]
\[a_{3}(x) = \frac{\alpha_2}{1 + x_1^{\gamma}}, \quad a_{4}(x) = x_2. \]
The values of rate constants in  \cite{Rathinam_2010} were $\alpha_1 = 50$, $\alpha_2 = 16$, $\beta = 2.5$ and $\gamma = 1$.
We consider each rate constant with above mentioned value $v$ to be a random variable with distribution U($0.8v, 1.2v$) and independent
of other constants.
Similarly as in  \cite{Rathinam_2010} the initial particle numbers of both species were set to $0$ and
the observable considered for sensitivity analysis was the number of particles of species $U$ at time $t=10$.
In Table \ref{tabGTS} we present estimates of sensitivity indices computed
from a 50000 step MC procedure using scheme E.
From Table \ref{tabGTS} we can see that the parameter with the greatest values of estimates of both main and total indices
for conditional mean is $\alpha_1$, while the lowest estimates of indices are these of parameter
$\beta$. 
In Table \ref{tabGTSGilRTC} we present estimates of variances of final MC estimators of procedures using RTC algorithm and GD method
with $2000$ steps for scheme E and $4000$ for EMe and ETe, so that the variances are computed for the same number of
process simulations used by certain schemes, similarly as in the previous Section. The estimates were obtained
from $50$ independent runs of each of these methods.
The estimated
variances are lower for RTC than GD method for all main and total indices, in agreement with
discussion in Section \ref{secVarDiff}. They are even about $4$ times lower for the indices associated with parameter $\beta$. 
Notice also that estimates of variances of final MC estimators given by 
scheme EM are not much lower than for scheme E for all main indices 
and similarly for schemes E and ET for total indices, which is close to theoretical bounds given by equalities in the left inequalities of   
relations (\ref{EMComp}) and (\ref{VitotComp}). 
\begin{table}[h]
\begin{tabular}{|l|c|c|c|c|c|c|}
\hline
$i $ & $\widetilde{V}_i$ & $\widetilde{V}_i^{tot} $ &$\widetilde{S}_i$ &$\widetilde{S}_i^{tot}$\\
\hline
$\alpha_1 $& $42.14 \pm 0.55$ & $43.33 \pm 0.56$ &$ 0.43 $&$ 0.45$\\ 
$\alpha_2 $& $12.83 \pm 0.40$ & $13.80 \pm 0.41$ &$ 0.13 $&$ 0.14$\\ 
$\beta $& $2.72 \pm 0.16$ & $2.73 \pm 0.15$ &$ 0.028 $&$ 0.028$\\ 
$\gamma $& $37.74 \pm 0.63$ & $38.54 \pm 0.64$ &$ 0.39 $&$ 0.4$\\ 
\hline
$i $ & $V_i$ & $V_i^{tot} $ & $S_i$ & $S_i^{tot}$ \\
\hline
$P $& $96.96 \pm 0.84$ & $246.7 \pm 1.0$ &$ 0.21 $&$ 0.53$ \\
$R $& $218.6 \pm 1.6$ & $368.4 \pm 1.6$ &$ 0.47 $&$ 0.79$ \\
$P,R$ & $465.3 \pm 1.7$ & $465.3 \pm 1.7$ &$ 1 $&$ 1$ \\
\hline
\end{tabular}
\caption{\label{tabGTS} Estimates of sensitivity indices for TS model computed in a 50000 step MC procedure
using RTC algorithm and scheme E.}
\end{table}
\begin{table}[h]
\resizebox{16cm}{!} {
\begin{tabular}{|l|c|c|c|c|c|c|}
\hline
{\multirow{3}{*}{i}}& \multicolumn{2}{c|}{E} & \multicolumn{2}{c|}{EMe}
& \multicolumn{2}{c|}{ETe} \\
\cline{2-7}
& GD & RTC & GD & RTC & GD & RTC \\
\cline{2-7}
 & \multicolumn{6}{c|}{$\widetilde{V}_i\err$} \\
\hline
$ \alpha_1$ &$16.05 \pm 0.14$ &$7.697 \pm 0.086$ &$13.820 \pm 0.089$ &$6.307 \pm 0.053$ &$85.42 \pm 0.28$ &$59.38 \pm 0.26$ \\
$ \alpha_2$ &$5.350 \pm 0.068$ &$4.067 \pm 0.050$ &$5.098 \pm 0.053$ &$3.636 \pm 0.034$ &$46.38 \pm 0.22$ &$39.61 \pm 0.18$ \\
$ \beta$ &$3.118 \pm 0.058$ &$0.640 \pm 0.022$ &$2.960 \pm 0.043$ &$0.586 \pm 0.016$ &$35.02 \pm 0.19$ &$14.387 \pm 0.100$ \\
$ \gamma$ &$11.41 \pm 0.11$ &$9.950 \pm 0.098$ &$9.733 \pm 0.074$ &$7.635 \pm 0.051$ &$67.52 \pm 0.28$ &$60.52 \pm 0.26$ \\
 \hline
 i & \multicolumn{6}{c|}{$\widetilde{V}_i^{tot}\err$} \\
 \hline
$ \alpha_1 $ &$16.27 \pm 0.15$ &$7.721 \pm 0.085$ &$60.39 \pm 0.23$ &$36.55 \pm 0.14$ &$13.755 \pm 0.082$ &$6.349 \pm 0.058$ \\
$ \alpha_2 $ &$5.720 \pm 0.075$ &$4.223 \pm 0.052$ &$39.96 \pm 0.16$ &$35.91 \pm 0.14$ &$5.463 \pm 0.054$ &$3.869 \pm 0.036$ \\
$ \beta $ &$3.311 \pm 0.058$ &$0.658 \pm 0.025$ &$34.22 \pm 0.19$ &$15.01 \pm 0.11$ &$3.188 \pm 0.053$ &$0.582 \pm 0.019$ \\
$ \gamma $ &$11.92 \pm 0.11$ &$10.315 \pm 0.099$ &$50.00 \pm 0.18$ &$46.25 \pm 0.17$ &$10.181 \pm 0.070$ &$7.942 \pm 0.061$ \\
 \hline
 i & \multicolumn{6}{c|}{${V}_i\err$} \\
 \hline
$P $& $22.39 \pm 0.14$ &$17.46 \pm 0.12$ &$28.83 \pm 0.15$ &$19.87 \pm 0.10$ &$28.55 \pm 0.14$ &$19.93 \pm 0.10$ \\
$P,R$ &$56.94 \pm 0.21$ &$68.88 \pm 0.27$ &$32.880 \pm 0.097$ &$38.28 \pm 0.10$ &$43.28 \pm 0.13$ &$45.95 \pm 0.12$\\
\hline
\end{tabular}
}
\caption{\label{tabGTSGilRTC} Estimates of variances of final MC estimators of sensitivity indices for
 TS model computed using RTC and GD methods for different schemes. See main text in Section \ref{secGTS} for details.}
\end{table}

\section{Many births - many deaths model}\label{secMBMD}
Chemical reaction network of many births - many deaths (MBMD) model contains one species $X$ and
$5$ different birth and death reactions can occur
\[\{R_{bi}:\ \emptyset \longmapsto  X,\quad R_{di}:\   X\longmapsto \emptyset \}_{\ i \in I_5}. \]
These reactions are ordered as follows
\begin{equation}
\{ R_i: = R_{bi},\ R_{5 + i}: = R_{di}\}_{i \in I_5}.
\end{equation}
The kinetic formulas of birth reactions are  $a_{bi}(x) = K_{bi}$ and of death reactions $a_{di}(x) = K_{di}x$, where
the distributions of kinetic rates are $K_{di} \sim \U(0.010,\ 0.040)$ and $K_{bi} \sim \U(0.10,\ 0.40)$ for $i \in I_5$.
The initial number of particles $C$ has distribution $U_d(5, 15)$ and we assume all parameters to be independent.
The observable for which we count sensitivity indices is the number of particles of species $X$ at time $t=5$.
The results of a 50000 step MC procedure computing sensitivity indices
 using RTC algorithm and scheme E are given in Table \ref{tabMBMD}.
In the next experiment we used three different constructions of DM to investigate variances of estimators of sensitivity indices using them.
The first two are RTC and GD constructions applied to the model described above.
 The third is GD construction but applied to a reaction network with reordered indices
\begin{equation}\label{newOrder}
\{ R_{2i-1}: = R_{bi},\ R_{2i}: = R_{di}\}_{i \in I_5}.
\end{equation}
The idea behind such reordering was to facilitate switching
between birth and death reactions in a given step of
GD construction for different values of model parameters and thus to reduce the coupling of the observable to the
noise term as discussed in Section \ref{secVarDiff}.
In Table \ref{tabCompMB}
we compare variances of final MC estimators of some sensitivity indices, estimated from 10 independent
runs of $5000$ step MC procedures using scheme E and three different constructions of DM we described,
similarly as in previous Sections.
We can see that the estiamtes of variances of final MC estimators of main sensitivity indices with respect to
parameters $K_{b1}$ and $K_{d1}$ estimated with RTC method are about two times lower than the ones estimated with GD method
with initial order of indices and over $30$ times lower for reordered indices. 
Note, however, that for main sensitivity indices of parameter $C$ estimator using GD method with reordered reaction indices achieves slightly 
lower variance than the one with initial order of indices. 
When using GD method in scheme EM we received lower estimates of variance for initial 
order of indices rather than the reordered indices for all main sensitivity indices of individual parameters
 and similarly for scheme ET and total sensitivity indices (data not shown).

\begin{table}[h]
\begin{tabular}{|l|c|c|c|c|c|c|}
\hline
$i $ & $\widetilde{V}_i$ & $\widetilde{V}_i^{tot} $ &$\widetilde{S}_i$ &$\widetilde{S}_i^{tot}$ \\
\hline
$X $& $2.8966 \pm 0.0038$ & $2.9240 \pm 0.0038$ &$ 0.73 $&$ 0.73$ \\ 
$k_{b1} $& $0.10384 \pm 0.00026$ & $0.10407 \pm 0.00026$ &$ 0.026 $ &$ 0.026$ \\ 
$k_{d1} $& $0.10476 \pm 0.00025$ & $0.11108 \pm 0.00025$ &$ 0.026 $&$ 0.028$ \\ 
\hline
$i$& $V_i$ & $V_i^{tot} $ & $S_i$ & $S_i^{tot}$ \\
\hline
$P $& $3.9792 \pm 0.0061$ & $6.0938 \pm 0.0068$ &$ 0.36 $&$ 0.55$ \\
$R $& $4.972 \pm 0.020$ & $7.087 \pm 0.021$ &$ 0.45 $&$ 0.64$ \\
$P,R$& $11.066 \pm 0.021$ & $11.066 \pm 0.021$ &$ 1 $&$ 1$ \\
\hline
\end{tabular}
\caption{\label{tabMBMD} Estimates of sensitivity indices for MBMD model computed in a 50000 step MC procedure
using RTC algorithm and scheme E.}
\end{table}

\begin{table}[h]
\begin{tabular}{|l|c|c|c|}
\hline
{\multirow{2}{*}{i}} & GDR & GDI & RTC \\
\cline{2-4}
& \multicolumn{3}{c|}{$\delta \widetilde{V}_i (10^{-6})$ } \\
\hline
$C $& $ 485.4 \pm 2.9 $&$ 528.6\pm 2.0 $ & $341.3 \pm 1.3 $  \\
$k_{b1}$ &$ 117.89 \pm 0.60 $ & $ 6.871 \pm 0.035 $& $ 3.441 \pm 0.014 $ \\
$k_{d1} $ &$ 95.43 \pm 0.36 $ &$ 5.282 \pm 0.046 $ &$ 2.725 \pm 0.030 $ \\
\hline
$i$& \multicolumn{3}{c|}{$\delta V_i (10^{-6})$} \\
\hline
$P $ & $1054.3 \pm 4.4 $& $ 688.3 \pm 3.3 $ & $ 632.5 \pm 3.8$\\
$P,R$ & $ 1356.1 \pm 4.4 $& $2104.0 \pm 7.4 $ & $ 2091.3 \pm 8.0 $\\
\hline
\end{tabular}
\caption{\label{tabCompMB} Estimates of variances of final MC estimators using GD algorithm with reordered of indices (GDR)
and GD and RTC methods with initial order of indices (denoted by GDI and RTC) for MBMD model. See main text in Section \ref{secMBMD} for details.}
\end{table}

\appendix

\chapter{\label{appMath} Mathematical background}
Unless stated otherwise we assume all measurable spaces to be standard Borel \cite{Ikeda1981}
and random variables to take values in such spaces, as well as all functions from one measurable space to another to be measurable \cite{Durrett}.
We say that random variables are equal or uniquely determined if they are equal or uniquely almost surely (a.s.). 
For a measure space $\mathcal{M}$ with measure $\mu$
by  $L^p(\mu)$ we denote the space of classes of equivalence on the set of functions $f$ from $\mathcal{M}$ to $\R$, such that
 $\int\ |f|^p \mathrm{d}\mu <\infty$ and $f \sim g$ iff $f= g$ a. s. $\mu$ (compare  \cite{rudin1970} Section 3.10),
 but as custom call these classes functions. 
For $p$, $q$ natural, $p > q$ and $\mu$ finite it holds that from $f\in L^p(\mu)$ it follows $f\in L^q(\mu)$.
Whenever dealing with random variables we assume some underlying probability space $(\Omega, \mathcal{F}, \PR)$ \cite{Durrett}.
\begin{defin}\label{defUD}
For two natural numbers $a \leq b$ we say that random variable $X$ has uniform discrete distribution $U_d(a,b)$ if for any $c$ natural such that
 $a \leq c \leq b$ 
\begin{equation}
P(X = c) = \frac{1}{b - a + 1}.
\end{equation}
\end{defin} 
\begin{defin}\label{defSupp}
 The support\cite{lehmann1998theory} of a probability measure $\mu$ on the measure space $(\R^n, \mathcal{B}(\R^n))$
is defined as set $\{ x\in R^n: P(A)>0 \text{ for every open rectangle $A$ containing }x $.
\end{defin}
For $A \in \mathcal{F}$ we denote by $\I_{A}$ the indicator of the set $A$, that is $\I_{A}(\omega) = 1$ if $\omega \in A$ and $0$ otherwise.
Definition of the conditional expectation of a random variable $Y\in L^1(\PR)$ given $X$, which we denote $\E(Y|X)$ can be found in any
standard probability \mbox{text-book}, like  \cite{Durrett}. Conditional expectation is a random variable and
is uniquely determined. 
We need the following well-known property of conditional expectation.
\begin{theorem}\label{indepCond}
For a function $f(X,Y) \in L(\PR)$ of independent random variables $X$, $Y$ we have
\begin{eqnarray}\label{condAveprop}
\E(f(X,Y)|X) = (\E f(x, Y))_{x = X}.
\end{eqnarray}
\end{theorem}
Conditional probability of an event $B \subset \Omega$ given random variable $X$ is defined as
\begin{equation}\label{PBX}
\PR(B|X) := \E(\I_{B}|X).
\end{equation}
Below we give definition of conditional distribution (\cite{borovkov1999mathematical} chap. 20, def. 1).
\begin{defin}\label{defMu}
For 2 random variables $X$ and $Y$ on $(\Omega, \mathcal{F}, \PR)$ and with values in $(S_1, \mathcal{S}_1)$ and $(S_2, \mathcal{S}_2)$
respectively, we call $\mu_{Y|X}:S_1 \times \mathcal{S}_2 \longmapsto [0,1]$ conditional distribution of $Y$ given $X$ if the following 
conditions are satisfied.
\begin{enumerate}
\item For every  $x \in S_1$  $\mu_{Y|X}(x, \cdot)$ is a probability measure on $\mathcal{S}_2$.
\item $\forall A \in \mathcal{S}_2 $ function $x \longmapsto \mu_{Y|X}(x , A) $ is measurable.
\item $\forall A \in \mathcal{S}_2 \quad \mu_{Y|X}(X , A) $ is a version of $\PR(Y \in A|X)$. 
\end{enumerate}
We  also say that $\mu_{Y|X}(,)$ is conditional distribution of $Y$ given $X=x$.
\end{defin}
  It turns out that for random variables $Y$ and $X$ with values in standard Borel spaces such as $(\R^n, \mc{B}(\R^n))$ 
  conditional distribution of $Y$ given $X$ exists and is in certain sense unique (see Chap. 1 in \cite{Ikeda1981}). 
It holds \cite{borovkov1999mathematical} that for $g(Y) \in L^1(\PR)$ and any random variable $X \sim \mu_X$
if $\mu_{Y|X}$ exists we have
\begin{equation}\label{condCond}
\E(g(Y)|X) = \int\! g(y)\, \mu_{Y|X}(X,dy).
\end{equation}
In particular, $\E(g(Y)|X)$ is certain function of $X$ and its distribution is determined by $\mu_X$ and $\mu_{Y|X}(,)$.
\begin{theorem}\label{aveSecFin}
Using notations and assumptions from Section \ref{genParSec}, if
for every value $p = (c,k)$ of $P=(C, K)$ we have \mbox{$h(p, R) \sim \mu_{DMCP}(RN(k), c)$}, then
$\mu_{DMCP}(RN(k),c)$ is conditional probability of $h(P,R)$ given $P=p$.
\end{theorem}
\begin{proof}
Point 1 in definition \ref{defMu} obviously holds, proof of point 2 is standard and the proof of point 3 is as follows
\begin{equation}
\begin{split}
\forall_{A \in \mathcal{B}(E^T)}\ \PR(h(P,R) \in A |P) &= \PR(h(p,R) \in A)_{p = P} \\
&= \mu_{DMCP}(RN(K), C)(A),\\
\end{split}
\end{equation}
where in the first equality we used (\ref{PBX}) and Theorem \ref{indepCond}.
\end{proof}

\chapter{\label{appFurthEstims} Some further estimators and proof of relation between inefficiency constants}
Use notations intoduced when defining scheme E in Section \ref{secMany}, 
we define here sub schemes of E for estimation of main and total sensitivity indices with respect to pairs $(P_i,P_j)$. 
We first define helper estimator 
\begin{equation}
\widehat{E^2}_{kl, E} := \frac{1}{8} \sum_{i=0}^{1} \sum_{j=0}^{1} (s_{k}[i][j]s_{k}[1-i][1-j] + s_{l}[i][j]s_{l}[1-i][1-j]).
\end{equation}
The estimator for $V_{(P_k,P_l)}$ is
\begin{equation}\label{VklEst}
\widehat{V}_{kl,E} := \frac{1}{4}\sum_{i=0}^{1} \sum_{j=0}^{1} (s_k[i][j] s_l[1-i][1-j]) - \widehat{E^2}_{kl,E},
\end{equation}
while for $\tilde{V}_{(P_k, P_l)}^{tot}$ the estimator is
\begin{equation}\label{vkltotEst}
\widehat{\widetilde{V}}^{tot}_{kl,E} := \frac{1}{8}\sum_{i=0}^{1} \sum_{j=0}^{1} (s_k[i][j] s_k[i][1 -j]
 + s_{l}[i][j] s_{l}[i][1-j] -  2 s_k[i][j]s_l[i][1-j]).
\end{equation}
The following lemma is needed for the proof of Theorem (\ref{thdEMEComp}). 
\begin{lemma}\label{thGBig}
For a random vector $X=(X_i)_{i=1}^3$ with independent coordinates let $g(X) \in L^4(\PR)$ 
and let random variables $Y_{i}[j]$ for  $i \in I_3$ and $j \in \{0,1\}$ 
be mutually independent and fulfill $Y_{i}[j] \sim X_i$. We denote, for $i,\ j,\ k \in \{0, 1\}$ 
\begin{equation}
g[i][j][k] = g(Y_{1}[i],Y_{2}[j],Y_{3}[k]). 
\end{equation}
For $i \in \{0,1\}$ we denote 
$B[i] = g[1-i][i][0]-g[i][i][0]$, 
$C[i] = g[1-i][1-i][1] - g[i][1-i][1]$
and 
\begin{equation}
A[i] = B[i]C[i].
\end{equation}
 It holds
\begin{equation}\label{covgeq0}
\Cov(A[0],A[1]) \geq 0.
\end{equation}
\end{lemma}
\begin{proof}
Let us denote by $(V_{J})_{J \subset I}$ the elements of ANOVA decomposition of variance of $g(X_1,X_2,X_3)$ (see section \ref{secANOVA}). 
From Theorem \ref{thCond} for every $i \in \{0,1\}$ $\E(A[i]) = 2V_1$. Thus we can write 
\begin{equation}\label{covBig}
\Cov(A[0],A[1]) = \E(B[0]B[1]C[0]C[1]) - 4V_1^2.
\end{equation}
Since $B[0]B[1] \sim C[0]C[1]$ it holds
\begin{equation}\label{bcbc}
\E(B[0]B[1]C[0]C[1]) = \Cov(B[0]B[1], C[0]C[1]) + \E^2(B[0]B[1]).
\end{equation}
From Theorem \ref{thCond} covariance on the rhs of (\ref{bcbc}) fulfills  
\begin{equation}\label{covsee}
\begin{split}
\Cov(B[0]B[1], C[0]C[1]) &= \Cov((g[1][0][0]-g[0][0][0])(g[0][1][0] - g[1][1][0]),\\
 &(g[1][0][1] -g[0][0][1])(g[0][1][1] -g[1][1][1]))\\
&= \Var(\E((g[1][0][0]-g[0][0][0])(g[0][1][0] - g[1][1][0])|Y_1,Y_2))\\ 
& =\Var(\E(B[0]B[1]|Y_1,Y_2)).
\end{split}
\end{equation}
We also have 
\begin{equation}\label{EB0B1}
\begin{split}
\E(B[0]B[1]) &= \E((g[1][0][0]-g[0][0][0])(g[0][1][0] - g[1][1][0])) \\
&= 2\E(\E^2(g[0][0][0]|Y_3)- \E^2(g[1][0][0]|Y_1,Y_3)) = -2(V_1 + V_{1,3}). \\
\end{split}
\end{equation}
Combining  (\ref{covBig}), (\ref{bcbc}), (\ref{covsee}) and (\ref{EB0B1}) we receive  
\begin{equation}
\Cov(A[0],A[1]) = \Var(\E(B[0]B[1]|Y_1,Y_2)) + 4V_{1,3}(2V_1 + V_{1,3}) \geq 0.
\end{equation}
\end{proof}

Below we provide the proof of Theorem \ref{thdEMEComp}.
\begin{proof}
For $i \in \{0,1\}$ let us denote
\begin{equation}\label{AiDef}
A[i] = (\wt{s}_k[i][0] - \wt{s}[i][0])(\wt{s}[1-i][1] -\wt{s}_k[1-i][1])
\end{equation}
and $k_{i}$ the number of function evaluations used by scheme $i$. Using our standard notation for observables of estimators we have
\begin{equation}\label{appkAkEM}
k_{E}\Var(V_{i,E}) = k_{EM}(\Var(V_{i,EM}) + \Cov(A[0], A[1])).
\end{equation}
For $X_1 = P_k$, $X_2 = P_{\sim k}$, $X_3 = R$, function $g$ such that $g(X_1,X_2,X_3) = f(P,R)$ and
$g[i][j][k]$ defined as in Lemma \ref{thGBig}, we have
$(\wt{s}_k[i][j],\wt{s}[i][j]) \sim (g[1-i][i][j], g[i][i][j])$ for $i,j \in \{0,1\}$. 
In particular $A[i]$ for $i \in \{0,1\}$ given by (\ref{AiDef}) and defined 
in Lemma \ref{thGBig} have the same joint distribution, thus here we also have
\begin{equation}\label{covA0AMgeq}
\Cov(A[0], A[1]) \geq 0.
\end{equation}
Expression (\ref{EMComp}) now follows from (\ref{appkAkEM}), (\ref{covA0AMgeq}) and Theorem \ref{thCov}.
\end{proof}

\chapter{\label{app3Params}More efficient estimators for 3 parameters}
Schemes for estimation of sensitivity indices of functions and their 
conditional expectations can be improved for the number of parameters $N_P$ equal to $3$, 
so that the new schemes allow for estimation of the same main and total sensitivity indices with respect to individual parameters, 
but with lower or equal inefficiency constants, in which equality holds only if both constants are equal to $0$. 
For certain index $k\in I_3$ 
one may resign from using $s_k$. The new estimates of the indices associated with $k$-th 
parameter are computed using values of observables $(\wt{P}_{(k)}[i])_{i=0}^1$ in place of 
$(\wt{P}[i])_{i=0}^1$ in estimator for computing $k$-th index. This does not change the 
expected value or variance of the estimator but allows for using $((s_l[i][j], s_m[1-i][j])_{i=0}^{1})_{j=0}^{1}$
in place of $((s[i][j], s_k[1-i][j])_{i=0}^1)_{j=0}^{1}$ in sub scheme computing $k$-th indices in scheme E  or 
$(s_l[i], s_m[1-i])_{i=0}^{1}$ in place of $(s[i], s_k[1-i])_{i=0}^1$ in such sub scheme of $O$. 
This reduces the number of function evaluations needed by a scheme by factor 
$\frac{3}{4}$, without changing the variance of its estimators. 
For instance the estimator for $k$-th main sensitivity index for output of such new scheme O3l created from O becomes 
\begin{equation}\label{Vl03l}
\widehat{V}_{k,O3k} := \frac{1}{2}(s_l[0] - s_m[1])(s_m[0] - s_l[1]),
\end{equation}
while for the total index of output 
\begin{equation}\label{Vltot03l}
\widehat{V}_{k,O3k}^{tot} := \frac{1}{4}\sum_{i=0}^1(s_l[i] - s_m[1-i])^2.
\end{equation}
For scheme E3l created in this way from E we have 
\begin{equation}
\widehat{V}_{l,E3k} := \frac{1}{4}\sum_{i=0}^1(s_l[i][0] - s_m[1-i][0])(s_m[1-i][1] - s_l[1-i][1])
\end{equation}
and
\begin{equation}
\widehat{\widetilde{V}}_{l,E3k}^{tot} := \frac{1}{4}\sum_{i=0}^1(s_l[i][0] - s_m[1-i][0])(s_l[i][1] - s_m[1-i][1]).
\end{equation}
Alternatively, instead of resigning from using $s_{k}$ one can apply it 
to compute analogous new estimators for $l$ and $m$ and use for all indices the averages
of old and new estimators. Thanks to theorem \ref{thAveVar} variances of such estimators are smaller or equal to the variances
of original estimators. 

\chapter{\label{appd}Sensitivity indices for simple birth model}
Instead of one birth process with rate equal to the sum of coordinates of random vector $K = (K_i)_{i=1}^3$ let us consider a model consisting 
of three birth processes with rates equal to its consecutive coordinates. This does not change conditional distribution of the process 
given the parameters, as 
in both cases it is the distribution of sum of three independent Poisson processes with given rates. In particular such change 
does not influence the values of variance-based sensitivity indices we compute here. 
We use a construction of such process resulting from integral equation (\ref{intEqu}) generalized to random parameters 
\begin{equation}\label{sumIndep}
Y_t = C + \sum_{i=1}^3 N_i(K_it).
\end{equation}
Using formula
\begin{equation}
 \sum_{i=1}^{n}i^3 = \frac{n(n+1)(2n+1)}{6}
\end{equation}
we receive for $X \sim U_d(a, b)$ (see Definition \ref{defUD})
\begin{equation}
\Var(X) = \frac{b(b+1)(2b+1) - (a-1)a(2a -1)}{6(b-a + 1)} - \left(\frac{a + b}{2}\right)^2.
\end{equation}
From the last expression we have
\begin{equation}
V_C = \Var(\E (Y_t|C)) = \Var(C) = 310.
\end{equation}
Denoting $\Pois(\lambda)$ Poisson distribution with parameter $\lambda$, 
for any Poisson process $N$ it holds $N(\lambda) \sim \Pois(\lambda)$. In particular 
\begin{equation}
 \E(N(\lambda)) = \lambda
\end{equation}
and 
\begin{equation}\label{Poiss2}
\E(N(\lambda)^2) = \lambda^2 + \lambda.
\end{equation}
From Theorem \ref{indepCond} we receive for $i \in I_3$ 
\begin{equation}\label{NCond}
\E(N_i(K_i t)|K_i) =  (\E(N_i(k_it)))_{k_i = K_i} = K_it.
\end{equation}
Since for $X \sim U(a, b)$ we have
\begin{equation}\label{varU}
\Var(X) = \frac{(b-a)^2}{12},
\end{equation}
we obtain 
\begin{equation}\label{VK1}
V_{K_1} = \Var(K_1 t) = 300,
\end{equation}
and similarly $V_{K_2} = 75$ and $V_{K_3} = 3$. From independence of summands in
\begin{eqnarray}
\E(Y_t|P) = \sum_{i=1}^3K_it + C
\end{eqnarray}
 we receive 
\begin{eqnarray}
V_P = V_C + \sum_{i=1}^3 V_{K_i} = 688
\end{eqnarray}
and $\tilde{V}_{i}^{tot} = V_{i}$ for every $i$-th parameter. 
From (\ref{Poiss2}) we receive for $i \in I_3$ 
\begin{eqnarray}\label{PoissSq}
\E(N(K_it)^2) = \E((\E(N(k_it)^2))_{k_i = K_i}) = \E(K_i^2)t^2 + \E(K_i)t,
\end{eqnarray}
while using further (\ref{VK1}) and (\ref{PoissSq}) 
\begin{equation}\label{NKVar}
\Var(N(K_it)) = \E(N(K_it)^2) - \E^2(K_it) = V_{K_i} + \E(K_i)t.
\end{equation}
From (\ref{NKVar}) and independence of summands in the rhs of (\ref{sumIndep}) we receive
\begin{equation}
\begin{split}
D & = \Var(Y_t) = V_C  + \sum_{i=1}^3 (V_{K_i} + \E(K_i)t)  \\
  & = V_P + \E (\sum_{i=1}^3K_i) t = 858.\\
\end{split}
\end{equation}
We also have
\begin{equation}
V_R^{tot} = D - V_P = 170.
\end{equation}


\bibliographystyle{plain}
\bibliography{pub}




\end{document}